\documentclass[pdflatex,sn-mathphys-num]{sn-jnl}
\usepackage{epstopdf}

\usepackage{graphicx}%
\usepackage{multirow}%
\usepackage{amsmath,amssymb,amsfonts}%
\usepackage{amsthm}%
\usepackage{mathrsfs}%
\usepackage[title]{appendix}%
\usepackage{xcolor}%
\usepackage{textcomp}%
\usepackage{manyfoot}%
\usepackage{booktabs}%
\usepackage{algorithm}%
\usepackage{algorithmicx}%
\usepackage{algpseudocode}%
\usepackage{listings}%
\usepackage{eucal}
\usepackage{comment}

\usepackage{tikz} 
\usetikzlibrary{shapes.geometric, arrows.meta, positioning} 

\usepackage{enumerate}
\usepackage{subfigure}

\theoremstyle{thmstyleone}%
\newtheorem{Thm}{Theorem}
\newtheorem{Lem}{Lemma}
\newtheorem{Def}{Definition}
\newtheorem{Cor}{Corollary}

\theoremstyle{thmstyletwo}%

\theoremstyle{thmstylethree}%

\begin{document}

\title[Article Title]{Defining a phylogenetic tree with the minimum number of small-state characters}


\author*[1]{\fnm{Yangjing} \sur{Long}}\email{yangjing@ccnu.edu.cn}
\author[1]{\fnm{Tong} \sur{Wang}}\email{wangtongtt1@163.com}



\affil[1]{\orgdiv{School of Mathematics and Statistics}, \orgname{Hubei Key Laboratory of Mathematical Physics}, \orgaddress{\street{Central China Normal University}, \city{Wuhan}, \postcode{430079}, \country{China}}}




\abstract{Phylogenetic trees represent evolutionary relationships and can be uniquely defined by sets of finite-state biological characteristics. Despite prior work showing that sufficiently large trees can be determined by $r$-state character sets, the minimal leaf thresholds $n_r$ remain largely unknown. In this work, we establish the 3-state case as $n_3 = 8$, providing a concrete base for higher-state analyses. We then resolve the 5-state problem by constructing a counterexample for $n=15$ and proving that for $n \geq 16$, $\lceil (n-3)/4 \rceil$ 5-state characters suffice to uniquely define any tree. Our approach relies on rigorous mathematical induction with complete verification of base cases and logically consistent inductive steps, offering new insights into the minimal conditions for character-based tree identification.
}




\maketitle

\section{Introduction}\label{sec1}

The construction of phylogenetic trees is a central task in evolutionary biology, aiming to infer both the evolutionary relationships among extant species and the ancestral trajectories that gave rise to them. As the fundamental input for tree reconstruction, the selection and number of characters—ranging from morphological traits to specific positions in genetic sequences—directly influence the accuracy of the inferred tree.

Within the character-based framework, theoretical approaches to tree reconstruction can be broadly classified into two categories, depending on how uniquely a tree is determined by a given character set: the weak-constraint approach based on \emph{identification}~\cite{Bordewich2006}, and the strong-constraint approach based on \emph{definition}~\cite{Semple2003}. This paper is grounded in the latter, more stringent setting, and investigates the conditions under which a minimal character set can uniquely define a tree.

We adopt the \emph{perfect phylogeny} framework, which assumes that each character evolves without homoplasy—that is, each state arises exactly once in the tree and never reverts or converges independently. This assumption corresponds to the combinatorial requirement that each character is convex on the tree, ensuring compatibility of the character set and enabling unique tree reconstruction~\cite{Semple2003}. The perfect phylogeny model forms the foundation of a rich line of research in theoretical phylogenetics.

A foundational result in this area, due to Buneman~\cite{Buneman1971}, characterizes when a dissimilarity map (i.e., a symmetric distance function on taxa) can be represented by a tree metric. The classical \emph{four-point condition} states that for any four taxa $a, b, c, d$, the two largest sums among
\[
\{d(a,b) + d(c,d),\ d(a,c) + d(b,d),\ d(a,d) + d(b,c)\}
\]
must be equal. Combinatorially, this condition ensures that each quartet of taxa admits exactly one of three unrooted binary tree topologies.
In this sense, the Buneman condition is equivalent to the constraint $n_2 = 3$, meaning that pairwise distances are sufficient to determine the correct tree topology when each character has only two states.

\smallskip

A central theoretical question in character-based phylogenetics is the following: given a phylogenetic tree $\mathcal{T}$ and a fixed number of character states $r$, what is the minimum number of convex $r$-state characters required to uniquely define the topology of $\mathcal{T}$? We refer to this as the \emph{Character Set Defining Phylogenetic Tree (CDT)} problem.

Within this framework, each character must be convex on the tree—meaning that each state appears exactly once and induces a connected subtree—thereby ruling out homoplasy and ensuring compatibility. This convexity condition corresponds to the perfect phylogeny model~\cite{Semple2003}, a foundational assumption in theoretical studies of evolutionary reconstruction. While it is known that allowing arbitrarily many character states permits small character sets to define any tree (e.g., $|\mathcal{C}| \leq 5$~\cite{Semple2002b}, later improved to $|\mathcal{C}| \leq 4$~\cite{Huber2005}), such models are biologically unrealistic. In most empirical settings, the number of states is limited (e.g., nucleotides, amino acids, or discrete morphological traits), and thus it is critical to understand how state limitations impact the expressive power of character sets.

To this end, Semple and Steel~\cite{Semple2003} established a general lower bound on the number of required characters:
\[
|\mathcal{C}| \geq \left\lceil \frac{n - 3}{r - 1} \right\rceil,
\]
where $n$ is the number of taxa (leaves) in $\mathcal{T}$ and $r$ is the number of character states. This inequality indicates that increasing the number of states allows fewer characters to suffice in defining the tree. Bordewich and Semple further demonstrated that this bound is asymptotically tight as $n$ grows large~\cite{Bordewich2015}.

Exact values of the \emph{critical leaf number} $n_r$ have been determined for small values of $r$:
\begin{itemize}
    \item For $r = 2$, the Buneman four-point condition implies that pairwise distances suffice, giving $n_2 = 3$~\cite{Buneman1971};
    \item For $r = 3$, recent work has established that $n_3 = 8$~\cite{wang2025}, We have also been informed (via personal communication) that a master’s student of Charles Semple obtained the same result independently;
    \item For $r = 4$, it is known that $n_4 = 13$~\cite{Bordewich2015}.
\end{itemize}

However, the case $r = 5$ has remained unresolved. The precise value of $n_5$—the minimum number of leaves such that there exists a set of convex five-state characters uniquely defining a phylogenetic tree—has not been previously determined. This paper resolves this open problem by constructing explicit examples and analyzing the associated combinatorial constraints. We prove that $n_5 = 16$, showing that for any fewer than 16 taxa, no such definition is possible, whereas 16 are sufficient.

\smallskip

Research on phylogenetic trees also involves many important directions, such as the compatibility of character sets~\cite{Buneman1971,Buneman1974,Steel1992}, the maximum parsimony method~\cite{Sourdis1988, Fischer2019, Fischer2023b}, and so on. In recent years, research in this field has been extended to the universal analysis of multi-state characters.

This paper focuses on the CDT problem under two fixed-state scenarios, $r = 3$ and $r = 5$, and aims to determine the exact value of the corresponding critical number of leaves $n_r$. While the result for $r = 3$ has been submitted to a separate Chinese-language publication~\cite{wang2025}, we include a full treatment of this case in Section~\ref{sec:3-state} to provide a unified and self-contained presentation. The case $r = 5$, which remains previously unresolved, is treated in Section~\ref{sec:5-state}.

\smallskip

The structure of the paper is as follows. Section~\ref{sec:pre} introduces relevant preliminary concepts. Section~\ref{sec:CDT} reviews prior work on the CDT problem and analyzes the topological and combinatorial properties associated with tree definability. Section~\ref{sec:3-state} focuses on the case $r = 3$, presenting a complete proof that $n_3 = 8$. Section~\ref{sec:5-state} addresses the case $r = 5$, establishing both lower and upper bounds and showing that $n_5 = 16$. Finally, Section~\ref{sec:conclusion} summarizes the main results and outlines future directions.

\section{Preliminary}\label{sec:pre}

A \emph{graph} $G$ is an ordered pair $(V, E)$, where $V$ is a non-empty set of \emph{vertices} and $E$ is a set of unordered pairs of elements from $V$ called \emph{edges}. We generally denote the vertex set and edge set of graph $G$ as $V(G)$ and $E(G)$, respectively. If $e = \{u, v\}$ is an edge of graph $G$, the vertices $u$ and $v$ are called the  \emph{ends} or \emph{end vertices} of $e$; vertices $u$ and $v$ are said to be \emph{adjacent}; and edge $e$ is \emph{incident} to vertices $u$ and $v$, meaning edge $e$ connects vertices $u$ and $v$. If two edges $e_1$ and $e_2$ of graph $G$ share a common end, they are said to be adjacent. An edge with coinciding ends is called a \emph{loop}, and edges connecting the same pair of vertices are called \emph{parallel edges}. A graph with neither loops nor parallel edges is generally called a \emph{simple graph}. All graphs studied in this paper are simple graphs.

Let $v$ be a vertex of $G$. The \emph{degree} of vertex $v$, denoted $d(v)$, is the number of edges in $G$ that are incident with $v$. A vertex of degree zero is called an \emph{isolated vertex}, and a vertex with degree one is called a \emph{pendant vertex}; an edge incident to a pendant vertex is called a \emph{pendant edge}. Let $e$ be an edge of $G$. Then $G \setminus e$ denotes the graph obtained by deleting $e$ from $E(G)$. $G \setminus v$ denotes the graph obtained by deleting vertex $v$ from $V(G)$ and all edges incident to $v$ from $E(G)$. If $V' \subseteq V$ and $E' \subseteq E$, $G \setminus E'$ denotes the graph obtained by deleting each edge in $E'$ from $G$, and $G \setminus V'$ denotes the graph obtained by deleting each vertex in $V'$ and their incident edges from $G$.

Let $H$ be a graph. If $V(H)$ and $E(H)$ are subsets of $V(G)$ and $E(G)$, respectively, then $H$ is called a \emph{subgraph} of $G$. If $V'$ is a non-empty subset of $V(G)$, 
the subgraph \emph{induced} by $V'$, denoted $G[V']$, is the subgraph with vertex set $V'$ and edge set consisting of all edges of $G$ with both ends in $V'$.

Two graphs $G_1 = (V_1, E_1)$ and $G_2 = (V_2, E_2)$ are said to be \emph{isomorphic} if there exists a bijective mapping $\psi: V_1 \rightarrow V_2$ that preserves adjacency; that is, $\{u, v\} \in E_1$ precisely if $\{\psi(u), \psi(v)\} \in E_2$. A \emph{path} in a graph $G$ is a sequence of distinct vertices $v_1, v_2, \ldots, v_k$ such that for each $i \in \{1, 2, \ldots, k-1\}$, $v_i$ is adjacent to $v_{i+1}$. Furthermore, if $v_1$ is adjacent to $v_k$, the subgraph of $G$ with the vertex set $\{v_1, v_2, \ldots, v_k\}$ and the edge set $\{v_k, v_1\} \cup \{\{v_i, v_{i+1}\} : i \in \{1, 2, \ldots, k-1\}\}$ is called a \emph{cycle}. A graph $G$ is called a \emph{chordal graph} if every induced subgraph that is a cycle has at most three edges. A graph $G$ is \emph{connected} if every pair of vertices is connected by a path. The maximal connected subgraphs of $G$ are called \emph{components} of $G$. The \emph{distance} between two vertices $u$ and $v$ in a connected graph $G$, denoted $d_G(u, v)$, is the number of edges in the shortest path connecting $u$ and $v$.

A connected acyclic graph is called a \emph{tree}. Let $T = (V, E)$ be a tree. A vertex with degree at most one is called a \emph{leaf}, and non-leaf vertices are called \emph{interior vertices}; an edge with both ends as interior vertices is called an \emph{interior edge}. We generally denote the sets of interior vertices and interior edges of tree $T$ as $\mathring{V}$ and $\mathring{E}$, respectively. A connected subgraph of $T$ is called a \emph{subtree}. A tree $T$ is called a \emph{binary tree} if all its interior vertices have degree three.
Let $v$ be a vertex of degree two in $T$. \emph{Suppressing the degree-two vertex} $v$ means deleting the edges incident to $v$ in $T$ and connecting the adjacent vertices of $v$. 
If $V'$ is a subset of $V$, the smallest connected subgraph of $T$ containing $V'$ is called the \emph{minimal subtree of $T$ induced by $V'$}, denoted $T(V')$.

An \emph{$X$-tree} $\mathcal{T}$ is an ordered pair $(T; \phi)$, where $T = (V, E)$ is a tree and $\phi$ is a mapping from $X$ to $V$, such that for all vertices $v \in V$ with degree at most two, $v \in \phi(X)$. An $X$-tree is also called a \emph{semi-labelled $X$-tree}. For two semi-labelled $X$-trees $\mathcal{T}_1 = (T_1; \phi_1)$ and $\mathcal{T}_2 = (T_2; \phi_2)$, where $T_1 = (V_1, E_1)$ and $T_2 = (V_2, E_2)$, they are said to be \emph{isomorphic} if there exists a unique adjacency-preserving bijection $\psi: V_1 \rightarrow V_2$ such that $\phi_2 = \psi \circ \phi_1$.

An \emph{$X$-split} is a partition of $X$ into two non-empty subsets, which can be regarded as a 2-state character on $X$. We denote a $X$-split with non-empty subsets $A$ and $B$ as $A|B$. Let $\mathcal{T} = (T; \phi)$ be an $X$-tree and $e$ be an edge of $T$. If $V_1$ and $V_2$ denote the vertex sets of the two components of $T \setminus e$, then $\phi^{-1}(V_1) | \phi^{-1}(V_2)$ is the $X$-split induced by the edge $e$, denoted $\sigma_e$. If $\mathcal{T}$ is a phylogenetic tree and $e$ is an interior edge, then $\sigma_e$ is called a \emph{nontrivial $X$-split}. A \emph{cherry} is a binary subset $\{x, y\}$ of $X$ where $x$ and $y$ are adjacent to the same interior vertex.

\section{Properties and Judgment Conditions of CDT Problem}\label{sec:CDT}

The construction methods of phylogenetic trees based on character sets mainly include two theoretical frameworks: ``identification-type" and ``definition-type", with this paper taking the ``definition-type" as the research foundation. In this context, this section systematically combs the basic properties and core theorems of the CDT (Character Set Define Phylogenetic Tree) problem from existing studies, focusing on the necessary and sufficient conditions and several sufficient conditions for character sets to define phylogenetic trees.

\subsection{Properties of CDT }\label{sec:CDT-prop}
In the phylogenetic tree research, the set of research objects consisting of modern species and taxonomic units is denoted $X$. There is a strict correspondence between the elements of this set and the leaves of phylogenetic trees: a \emph{phylogenetic $X$-tree $\mathcal{T}$} can be formally defined as an ordered pair $(T; \phi)$, where $T$ is a tree structure without degree-two vertices, and $\phi$ is a bijective mapping from set $X$ to the leaf set of $T$. 
If, in addition, every interior vertice of $T$ has degree three, $\mathcal{T}$ is a \emph{binary phylogenetic $X$-tree}.
 Due to its strict dichotomous nature, the binary phylogenetic tree structure can accurately simulate the branching events in species evolution and has become the mainstream modeling framework for molecular evolution and phylogenetic reconstruction. All phylogenetic trees studied in this paper are unrooted binary phylogenetic trees.

\begin{figure}[H]
    \centering
    \includegraphics[width=0.65\textwidth]{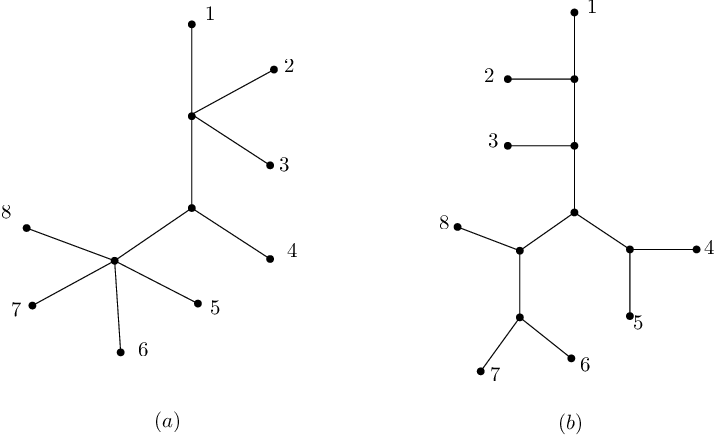}
    \caption{(a) Phylogenetic $X$-tree. (b) Binary phylogenetic $X$-tree.}
    \label{fig:example_phylogenetic}
\end{figure}

A \emph{character} on $X$ is a function from $X$ to a state set $C$, typically denoted as $\chi$. If the number of states $|\chi(X)| = r$, then $\chi$ is called an \emph{$r$-state character.} For $\alpha \in \chi(X)$, $\mathcal{T}(\alpha)$ denotes the minimal subtree of $\mathcal{T}$ containing the leaves with state $\alpha$. The convexity of characters on phylogenetic trees is an important concept used to describe the distribution patterns of certain attributes of species or genes during evolution. A character $\chi$ is said to be \emph{convex} on $\mathcal{T}$ if the vertex sets of the trees in the set $\{\mathcal{T}(\alpha) : \alpha \in \chi(X)\}$ are pairwise disjoint. If every character in $\mathcal{C}$ is convex on $\mathcal{T}$, then $\mathcal{C}$ is said to be convex on $\mathcal{T}$, and in this case, $\mathcal{C}$ is called \emph{compatible}.

\begin{Def}
Let $\mathcal{T}$ be an $X$-tree. A character set $\mathcal{C}$ is said to \emph{define} $\mathcal{T}$ if $\mathcal{C}$ is convex on $\mathcal{T}$, and any other $X$-tree $\mathcal{T}'$ on which $\mathcal{C}$ is convex is isomorphic to $\mathcal{T}$. In other words, $\mathcal{T}$ is the unique $X$-tree (up to isomorphism) on which $\mathcal{C}$ is convex. We also say that $\mathcal{T}$ is defined by $\mathcal{C}$.
\end{Def}

To simplify the representation of characters, we use $\pi(\chi)$ to denote the partition of $X$ induced by $\chi$, i.e., $\pi(\chi) = \{\chi^{-1}(\alpha_i) : \alpha_i \in \chi(X)\}$. In the above examples, $\pi(\chi_1) = \{\{1, 2\}, \{3, 4\}, \{5, 6\}\}$ and $\pi(\chi_2) = \{\{1, 2, 3\}, \{4, 5, 6\}\}$.

Next, we introduce several theorems related to the properties of phylogenetic trees.

\begin{Thm}\emph{\cite{Bordewich2015}}\label{basic_thm_conclusion_1}
Let $r$ be a positive integer greater than one. For any phylogenetic tree $\mathcal{T}$ with $n$ leaves (where $n \geq n_r$, and $n_r$ is a positive integer threshold related to $r$), there exists a set $\mathcal{C}$ of $r$-state characters satisfying the following conditions:
\begin{enumerate}
  \item  Uniqueness: $\mathcal{C}$ defines $\mathcal{T}$;
  \item  Set size: The number of characters required is $|\mathcal{C}| = \lceil \frac{n-3}{r-1} \rceil$.
\end{enumerate}
\end{Thm}

An \emph{internal subtree} of $\mathcal{T}$ is a subtree whose edges are all internal. Up to isomrphism, the 6-leaf tree in Fig.~\ref{fig:n=12}($b$) called the \emph{snowflake}, the 3-leaf tree in Fig.~\ref{fig:n=12}($c$) called the \emph{3-star}.

\begin{Thm}\emph{\cite{Huber2024}}\label{thm_snowflake}
Let $\mathcal{T}$ be a  phylogenetic tree. Then $\mathcal{T}$ is defined by a set of at most three characters if and only if $\mathcal{T}$ has no internal subtree isomorphic to the snowflake.
\end{Thm}

According to Theorem \ref{thm_snowflake}, the phylogenetic tree with 12 leaves in Fig. \ref{fig:n=12} contains an internal subtree isomorphic to a snowflake, so it cannot be defined by a set of three characters.

\begin{figure}[H]
  \centering
  \subfigure[]
  {
   \begin{minipage}[b]{.4\linewidth}
     \centering
     
     \includegraphics[width=0.85\textwidth]{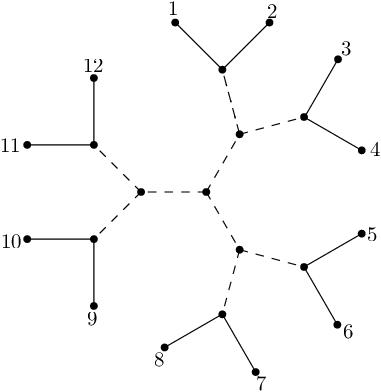}
   \end{minipage}
  }
  \subfigure[]
  {
   \begin{minipage}[b]{.3\linewidth}
     \centering
     \includegraphics[width=0.85\textwidth]{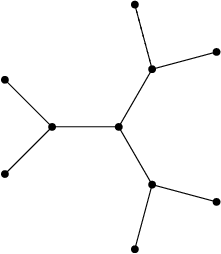}
   \end{minipage}
  }
  \subfigure[]
  {
    \begin{minipage}[b]{.2\linewidth}
     \centering
     
     \includegraphics[width=0.55\textwidth]{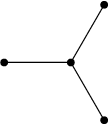}
   \end{minipage}
  }
  \caption{($a$) A phylogenetic tree $\mathcal{T}$ on $\{1, 2, \ldots, 12\}$. ($b$) The snowflake. ($c$) The 3-star.}
  \label{fig:n=12}
\end{figure}

\begin{Thm}\emph{\cite{Huber2024}}\label{thm_3-star}
Let $\mathcal{T}$ be a  phylogenetic tree. Then $\mathcal{T}$ is defined by a set of at most two characters if and only if $\mathcal{T}$ has no internal subtree isomorphic to the 3-star.
\end{Thm} 
The phylogenetic tree studied in Theorem \ref{thm_3-star} has a special topological structure, which conforms to the definition of a \emph{caterpillar}. This structure can be intuitively explained through Fig \ref{fig:caterpillar}: as shown in the figure, the main trunk of the tree is composed of continuous internal vertices, and each internal nvertice is directly connected to at least one leaf, forming a morphology similar to the body segments and ventral feet of a caterpillar.

\begin{figure}[H]
\centering
\includegraphics[width=0.45\linewidth]{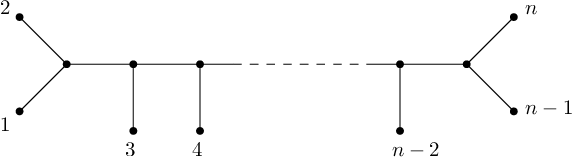}
\caption{A caterpillar on \(\{1, 2, \ldots, n\}\).}
\label{fig:caterpillar}
\end{figure}

\subsection{Necessary and Sufficient Conditions for CDT}\label{sec:CDT-cy}
This section introduces a necessary and sufficient condition for CDT (Character Set Defining Phylogenetic Tree), which can be used to determine whether a character set $\mathcal{C}$ defines a phylogenetic tree $\mathcal{T}$. Several auxiliary definitions introduced below (such as discrimination, Int($\mathcal{C}$), etc.) will be specifically used to support the proof and application of Theorem \ref{basic_thm_judgment}. Therefore, this section elaborates on them simultaneously to ensure the completeness of the theoretical tools.

Let $\mathcal{T}$ be a phylogenetic $X$-tree, $\mathcal{C}$ be a character set on $X$, and $e = \{u, v\}$ be an interior edge of $\mathcal{T}$. If there exists a character $\chi$ in $\mathcal{C}$ and $A_1, A_2 \in \pi(\chi)$ such that the vertex set of the minimal subtree of $\mathcal{T}$ containing $A_1$ includes $u$ but not $v$, and the vertex set of the minimal subtree of $\mathcal{T}$ containing $A_2$ includes $v$ but not $u$, then $e$ is said to be \emph{distinguished by $\chi$}. If each interior edge of $\mathcal{T}$ is distinguished by a character in $\mathcal{C}$, then $\mathcal{T}$ is said to be \emph{distinguished by $\mathcal{C}$}. For example, in Fig.~\ref{fig:Int(C)-example}, $e_1$ is distinguished by the character $\chi_1$, where $\pi(\chi_1) = \{\{1, 2\}, \{3, 4, 5, 6\}, \{7, 8\}\}$.

\begin{figure}[H]
\centering
\subfigure[]
  {
   \begin{minipage}[b]{.4\linewidth}
     \centering
     \includegraphics[width=0.65\textwidth]{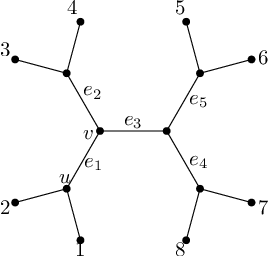}
   \end{minipage}
  }\subfigure[]
  {
   \begin{minipage}[b]{.45\linewidth}
     \centering
     \includegraphics[width=0.9\textwidth]{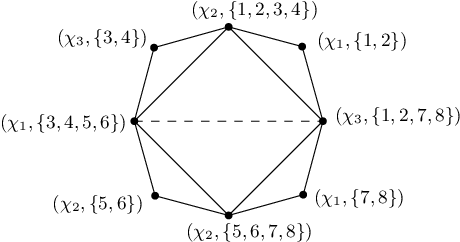}
   \end{minipage}
  }
\caption{($a$) A phylogenetic tree on $\{1,2,\ldots,8\}$. ($b$) The solid lines represent Int($\mathcal{C}$), and the entire graph is the completely chordal graph restricted by Int($\mathcal{C}$).}
\label{fig:Int(C)-example}
\end{figure}

Let \emph{$\text{Int}(\mathcal{C})$} denote the partition intersection graph of $\mathcal{C}$, that is, the graph with vertex set $\{(\chi, A) : \chi \in \mathcal{C}, A \in \pi(\chi)\}$ and edge set $\{\{(\chi, A), (\chi', B)\} : A \cap B \neq \emptyset\}$ (where $\emptyset$ denotes the empty set), where $\chi \neq \chi'$ if there is an edge between $(\chi, A)$ and $(\chi', B)$. That is, when $\chi \neq \chi'$ and $A$ and $B$ have common elements, $(\chi, A)$ and $(\chi', B)$ are connected by an edge. The \emph{restricted chordal completion $G$  of $\text{Int}(\mathcal{C})$} is a chordal graph that is obtained from $\text{Int}(\mathcal{C})$ by adding only edges that join vertices whose first components are distinct. Edges of $G$ that are not in $\text{Int}(\mathcal{C})$ are called \emph{completion edges}. Furthermore, if $G$ becomes a non-chordal graph after deleting any completion edge, then $G$ is a \emph{minimal restricted chordal completion graph of $\text{Int}(\mathcal{C})$}.

\begin{Thm}\emph{\cite{semple2002a}}\label{basic_thm_judgment}
Let $\mathcal{T}$ be a phylogenetic $X$-tree and $\mathcal{C}$ be a set of characters on $X$. Then $\mathcal{C}$ defines $\mathcal{T}$ if and only if
\begin{enumerate}
  \item Each character is convex on $\mathcal{T}$, and each interior edge of $\mathcal{T}$ is distinguished by a character in $\mathcal{C}$;
  \item $\text{Int}(\mathcal{C})$ has a unique minimal restricted chordal completion.
\end{enumerate}
\end{Thm}

\subsection{Sufficient Conditions for CDT} \label{sec:CDT-cf}
This section introduces sufficient conditions for CDT (Character Sets Define Phylogenetic Trees), and further proposes and demonstrates a set of alternative determination conditions that do not rely on chordal graph operations.

\begin{Thm}\emph{\cite{Buneman1971}}\label{basic_thm_judgment_2}
Let $\mathcal{T}$ be a phylogenetic $X$-tree and $\sum$ be the set of non-trivial $X$-splits of $\mathcal{T}$. Then $\sum$ defines $\mathcal{T}$, that is, up to isomorphism, $\mathcal{T}$ is the unique phylogenetic $X$-tree whose set of $X$-splits includes $\sum$.
\end{Thm}
For a compatible collection $\mathcal{C}$ of characters on $X$, we say that $\mathcal{C}$ \emph{infers} a character $\chi$ if $\chi$ is convex on every $X$-tree on which $\mathcal{C}$ is convex. Combining Theorem \ref{basic_thm_judgment_2}, if the character set $\mathcal{C}$ can infer all nontrivial $X$-splits of $\mathcal{T}$, then $\mathcal{C}$ defines $\mathcal{T}$.

We first discuss an equivalent definition of a character $\chi$ being convex on an $X$-tree $\mathcal{T}$: if there exists a subset $F$ of edges of $\mathcal{T}$ such that the graph obtained by deleting $F$ from $\mathcal{T}$ satisfies the following property: for all $A, B \in \pi(\chi)$, there are two connected components in this graph, where $\phi(A)$ is a subset of the vertex set of one component and $\phi(B)$ is a subset of the vertex set of the other component. In this case, we say that $\chi$ \emph{is diaplayed by} $F$.

Let $\mathcal{T} = (T, \phi)$ be an $X$-tree and $F$ be a set of some edges of $T$. Let $V_1, V_2, \ldots, V_k$ be the vertex sets of the components of $T \setminus F$. The \emph{partition of $X$ displayed by $F$} is $\{\phi^{-1}(V_i) : i \in \{1, 2, \ldots, k\}\}$.

\begin{Lem}\emph{\cite{Bordewich2015}}\label{basic_lem_1}
Let $r \geq 2$ and let $\mathcal{T}$ be a phylogenetic $X$-tree. Furthermore suppose that $\mathcal{T}$ has a path containing(in order) $2r - 2$ interior edges $e_1, e_2, \ldots, e_{2r - 2}$. Let $\{X_1, X_2, \ldots, X_{2r - 1}\}$ be the partition of $X$ displayed by $E' = \{e_1, e_2, \ldots, e_{2r - 2}\}$, where, for all $i \in \{1, 2, \ldots, 2r - 2\}$, the edge $e_i$ is the only edge in $E'$ in the  minimal subtree of $\mathcal{T}$ connecting the elements in $X_i \cup X_{i + 1}$. Then any two $r$-state characters $\chi_1$ and $\chi_2$ with
\begin{center}
  \[
\pi(\chi_1) = \{X_1, X_2 \cup X_3, X_4 \cup X_5, \ldots, X_{2r - 2} \cup X_{2r - 1}\}\]  
\end{center}
and
\begin{center}
\[\pi(\chi_2) = \{X_1 \cup X_2, X_3 \cup X_4, \ldots, X_{2r - 3} \cup X_{2r - 2}, X_{2r - 1}\}.
\]
\end{center}
infer $X$-splits $\sigma_{e_1}, \sigma_{e_2}, \ldots, \sigma_{e_{2r - 2}}$.
\end{Lem}

Let $e$ be an edge of an $X$-tree $\mathcal{T} = (T, \phi)$, and let $V_1, V_2$ be the vertex sets of the components of $T \setminus e$. Let $\chi_e$ denote the character $\chi_e : X \to \{\alpha_e, \beta_e\}$ defined, for  all $y \in X$, by
\[
\chi_e(y) = 
\begin{cases} 
\alpha_e & \text{if } y \in \phi^{-1}(V_1), \\
\beta_e & \text{if } y \in X \setminus \phi^{-1}(V_1).
\end{cases}
\]

\begin{Lem}\emph{\cite{Bordewich2015}}\label{basic_lem_2}
Let $\mathcal{T} = (T; \phi)$ be an $X$-tree and let $\chi$ be a character on $X$ that is convex on $\mathcal{T}$, where $\pi(\chi) = \{Y_1, Y_2, \ldots, Y_r\}$, where $r\geq 2$. Let $\{f_1, f_2, \ldots, f_{r-1}\}$ be a set of edges that displays $\chi$. Let $E' = \{e_1, e_2, \ldots, e_s\}$ be a subset of edges of $\mathcal{T}$ with $E' \cap \{f_1, f_2, \ldots, f_{r-1}\}$ empty satisfying the following two properties:
\begin{enumerate}
  \item for all distinct $i, j \in \{1, 2, \ldots, r-1\}$, there is an interior edge $e \in E'$ on the path from an end vertex of $f_i$ to an end vertex of $f_j$; and
  \item for each $e = \{u, v\} \in E'$, there is a path from $u$ (resp., $v$) to a vertex $w$ of $\mathcal{T}$ avoiding $v$ (resp., $u$) and  $f_1, f_2, \ldots, f_{r-1}$, and $\phi^{-1}(w)$ is nonempty.
\end{enumerate}
Then the collection 
\begin{center}
  $\{\chi, \chi_{e_1}, \chi_{e_2}, \ldots, \chi_{e_s}\}$   
\end{center}
of characters on $X$ infers each of the $X$-splits $\sigma_{f_1}, \sigma_{f_2}, \ldots, \sigma_{f_{r-1}}$.
\end{Lem}

\begin{Def}\label{Def_restriction}
Let $\mathcal{T}'$ be a phylogenetic $X'$-tree and $\mathcal{T}$ be a phylogenetic $X$-tree with $X \subseteq X'$. If $\mathcal{T}$ can be obtained from he minimal subtree of $\mathcal{T}'$ connecting the elemnts in $X$ by suppressing degree-two vertices, then $\mathcal{T}$ is said to be \emph{a restriction of $\mathcal{T}'$}.
\end{Def}
Let $A'|B'$ be an $X'$-split and $A|B$ be an $X$-split.  
If for some choices of $A$ and $B$, $A \subseteq A'$ and $B \subseteq B'$, then $A'|B'$ is said to \emph{extend $A|B$}. Note that $\mathcal{T}$ is a restriction of $\mathcal{T}'$ if and only if for each interior edge $e$ of $\mathcal{T}$, there exists an edge $f$ in $\mathcal{T}'$ such that the $X'$-split corresponding to $f$ extends the $X$-split corresponding to $e$.

\begin{Def}\label{Def_representative}
Suppose $\mathcal{T}$ is a restriction of $\mathcal{T}'$ and $F$ is a set of some interior edges of $\mathcal{T}'$. If for each interior edge $e$ of $\mathcal{T}$, there is exactly one edge $f$ in $F$ such that $\sigma_f$ extends $\sigma_e$, then $F$ is called a \emph{$\mathcal{T}$-representable subset}.
\end{Def}

Let the character set $\mathcal{C}$ on $X$ be convex on $\mathcal{T}$, and $\chi$ be a character in $\mathcal{C}$ displayed by $E_{\chi}$. The character on $X'$ displayed by the set $\{f : \sigma_f \text{ extends } \sigma_e, e \in E_{\chi}, f \in F\}$ in $\mathcal{T}'$ is denoted as $\chi_F$. At this time, $\chi_F$ is a character on $X'$. Furthermore, denote $\mathcal{C}_F = \{\chi_F : \chi \in \mathcal{C}\}$.

\begin{Lem}\emph{\cite{Bordewich2015}}\label{basic_lem_3}
Let $\mathcal{T}'$ be a phylogenetic $X'$-tree and $\mathcal{T}$ be a phylogenetic $X$-tree with $X \subseteq X'$. Suppose $\mathcal{T}$ is a restriction of $\mathcal{T}'$. Let $F$ be a $\mathcal{T}$-representable subset of interior edges of $\mathcal{T}'$. If $\mathcal{C}$ is a collection of characters on $X$ that defines $\mathcal{T}$, then the collection $\mathcal{C}_F$ of characters on $X'$ infers each of the $X'$-splits of $\mathcal{T}'$ induced by the edges in $F$.
\end{Lem}



\section{Phylogenetic Tree Definition by 3-State Characteristic Sets}~\label{sec:3-state}

This section focuses on the special case where the number of character states is fixed at $r = 3$. We determine the minimal number of leaves $n_3$ such that every phylogenetic tree with $n \geq n_3$ leaves can be uniquely defined by a set of 3-state characters, where the size of the set is \(\left\lceil \frac{n - 3}{2} \right\rceil\).


\begin{Thm}\label{main_thm_r=3}
Let $\mathcal{T}$ be a phylogenetic $X$-tree, where $n = |X|$.
If $n \geq 8$, then there is a collection $\mathcal{C}$ of 3-state characters of size 
\begin{center}
  $|\mathcal{C}| = \lceil \frac{n - 3}{2} \rceil$  
\end{center}
 that defines $\mathcal{T}$.
Moreover, if $n = 7$, there is a phylogenetic $X$-tree that is not  defined by any collection of $2 = \lceil \frac{7 - 3}{2} \rceil$ characters.
\end{Thm}

\subsection{Counterexample for $n = 7$}

We construct a 7-leaf phylogenetic tree that contains an internal subtree isomorphic to the 3-star (Fig.~\ref{fig:3star-counterexample}). By Theorem \ref{thm_3-star}, this structure cannot be defined by two characters of any arity. In particular, it cannot be defined by two 3-state characters.

\begin{figure}[H]
\centering
\includegraphics[width=0.45\linewidth]{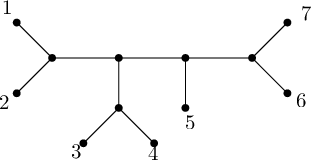}
\caption{A 7-leaf phylogenetic tree containing a 3-star as a subtree. This tree cannot be defined by two 3-state characters.}
\label{fig:3star-counterexample}
\end{figure}

\subsection{\texorpdfstring{Base cases: $n = 8$ and $n = 9$}{Base cases: n = 8 and n = 9}}

The remaining proof of Theorem \ref{main_thm_r=3} proceeds by induction on \( n \). For \( n = 8 \), let \( F_1, F_2, F_3 \) be a partition of the internal edges of $\mathcal{T}$ such that \( |F_1|, |F_2|, |F_3| \leq 2 \). For each \( i \in \{1, 2, 3\} \), let \( \chi_i \) be a character in $\mathcal{C}$ displayed by \( F_i \). If \( T \) is a caterpillar tree, the internal edges of $\mathcal{T}$ lie on a single path; assign the last internal edge on this path to \( F_3 \), and distribute the other four internal edges evenly into \( F_1 \) and \( F_2 \) such that no two edges in the same set are adjacent. If $\mathcal{T}$ is not a caterpillar tree, $\mathcal{T}$ has a longest path with at most two internal edges not on this path; assign the internal edges not on the longest path to \( F_3 \), and distribute the remaining internal edges evenly into \( F_1 \) and \( F_2 \) such that no two edges in the same set are adjacent. Construct \( \text{Int}(\{\chi_1, \chi_2, \chi_3\}) \) accordingly; by Theorem \ref{basic_thm_judgment}, \( \{\chi_1, \chi_2, \chi_3\} \) defines $\mathcal{T}$.

For \( n = 9 \), similar to the above discussion: if $\mathcal{T}$ is a caterpillar tree, its internal edges lie on a single path; assign the 1st and 4th internal edges on this path to \( F_1 \), the 2nd and 5th to \( F_2 \), and the 3rd and 6th to \( F_3 \). If \( T \) is not a caterpillar tree, $\mathcal{T}$ has a longest path with at most two internal edges not on this path. If there are two internal edges not on the longest path, assign them to \( F_3 \), and distribute the remaining internal edges evenly into \( F_1 \) and \( F_2 \) such that no two edges in the same set are adjacent. If there is only one internal edge not on the longest path, assign this edge together with either the first or last internal edge on the longest path to \( F_3 \), and distribute the remaining internal edges evenly into \( F_1 \) and \( F_2 \) such that no two edges in the same set are adjacent. Construct \( \text{Int}(\{\chi_1, \chi_2, \chi_3\}) \) accordingly; by Theorem \ref{basic_thm_judgment}, \( \{\chi_1, \chi_2, \chi_3\} \) defines $\mathcal{T}$.

\subsection{\texorpdfstring{Inductive step for $n \geq 8$}{Inductive step for n ≥ 8}}
For \( n \geq 8 \), the phylogenetic tree $\mathcal{T}$ must contain two cherries separated by at least three internal edges. Without loss of generality, assume these cherries are labeled \( \{n-3, n-1\} \) and \( \{n-2, n\} \). Let \( \mathcal{T}_{n-2} \) be the phylogenetic tree obtained by removing the leaves labeled \( n-1 \) and \( n \) from $\mathcal{T}$ and contracting degree-2 vertices. By induction, there exists a set \( \mathcal{C}_{n-2} \) of 3-state characters such that \( \mathcal{C}_{n-2} \) defines \( \mathcal{T}_{n-2} \), where \( |\mathcal{C}_{n-2}| = \lceil \frac{n-5}{2} \rceil \). Here, \( \mathcal{T}_{n-2} \) is said to be a restriction of \( \mathcal{T} \), and the set of internal edges \( F \) of \(\mathcal{T}\) is called a \( \mathcal{T}_{n-2} \)-representable subset. By Lemma \ref{basic_lem_3}, the set of characters \( (\mathcal{C}_{n-2})_F \) on \( \{1, 2, \ldots, n\} \) can infer the \( X \)-splits induced by the edges in \( F \). Let \( \chi \) be a character on \( \{1, 2, \ldots, n\} \) with \( \pi(\chi) = \{\{n-3, n-1\}, \{n-2, n\}, \{1, 2, \ldots, n-4\}\} \). Let \( \mathcal{C} = (\mathcal{C}_{n-2})_F \cup \{\chi\} \); then \( \mathcal{C} \) is a set of 3-state characters on \( \{1, 2, \ldots, n\} \) with \( |\mathcal{C}| = \lceil \frac{n-3}{2} \rceil \). Since the cherries \( \{n-3, n-1\} \) and \( \{n-2, n\} \) are separated by at least three internal edges, by Lemma~\ref{basic_lem_2}, \( \mathcal{C} \) defines \( \mathcal{T} \).

\qed

\section{ Phylogenetic Tree Definition by 5-State Characteristic Sets}\label{sec:5-state}
This chapter deeply explores the minimum number of leaves threshold $n_5$ for phylogenetic trees defined by 5-state character sets, and finally establishes a rigorous proof of Theorem~\ref{main_thm} ($n_5 = 16$).

Section \ref{sec:proof-idea} briefly describes the proof idea. Sections \ref{sec:15-leaf} and \ref{sec:16-19-leaf} elaborate on the phylogenetic trees with 16 to 19 leaves in the basic cases. Section \ref{sec:Y-like} investigates a special type of phylogenetic tree (Y-like trees) to prepare for the induction step. Finally, Section \ref{sec:main-proof} presents the complete proof process.

\subsection{Proof Idea of Theorem~\ref{main_thm}}
\label{sec:proof-idea}

\begin{Thm}\label{main_thm}
Let $\mathcal{T}$ be a phylogenetic $X$-tree, where $n = |X|$.
If $n \geq 16$, then there is a collection $\mathcal{C}$ of 5-state characters of size 
\begin{center}
  $|\mathcal{C}| = \lceil \frac{n - 3}{4} \rceil$  
\end{center}
 that defines $\mathcal{T}$.
Moreover, if $n = 15$, there is a phylogenetic $X$-tree that is not  defined by any collection of $3 = \lceil \frac{15 - 3}{4} \rceil$ characters.
\end{Thm}

Part of the proof of Theorem~\ref{main_thm} can be illustrated by constructing specific counterexamples. When the number of leaves $n = 15$, there exists a phylogenetic tree that cannot be defined by three 5-state characters. As shown in Fig.~\ref{fig:internal_path_5}, this phylogenetic tree has a special topological structure: its dashed internal subtree is isomorphic to a snowflake graph. Based on this structural feature, according to the necessary conditions of Theorem~\ref{thm_snowflake} on character sets defining tree structures, it can be inferred that this tree structure cannot be defined by a set of three characters.

\begin{figure}[H]
  \centering
  \includegraphics[width=0.45\columnwidth]{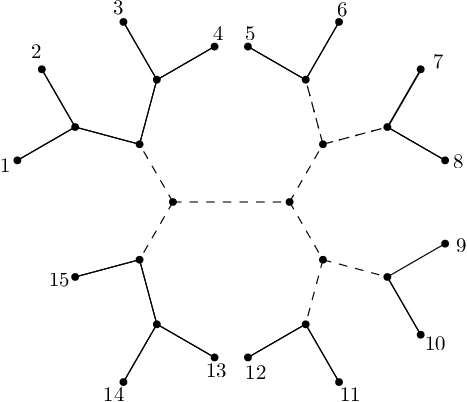}\\
  \caption{Phylogenetic tree $\mathcal{T}$ on $\{1,2,\ldots,15\}$, which cannot be defined by three 5-state characters.}
  \label{fig:internal_path_5}
\end{figure}

The complete proof of Theorem~\ref{main_thm} uses mathematical induction, with its core difficulty lying in the construction of base cases.
Specifically, when the number of taxa $n = 16, 17, 18$ (Corollary~\ref{pre_n=16_17_18}) and $n = 19$ (Corollary~\ref{pre_n=19_all}), the validity of the proposition needs to be verified one by one. Given the high complexity of the topological structures of phylogenetic trees in these cases, this paper adopts a step-by-step strategy: first, systematically study specific phylogenetic trees with 15 leaves—these trees need to satisfy two structural constraints simultaneously, and then generalize the conclusions to larger base cases through Lemma~\ref{basic_lem_2} and Lemma~\ref{basic_lem_3}. This method significantly reduces the proof complexity.
In the analysis of specific phylogenetic trees with 15 leaves, we systematically classify the tree structures into four mutually exclusive cases according to a key parameter: the number of interior edges in the longest path. By separately examining the character compatibility of each type of structure, we finally establish the mathematical foundation for inductive recursion.

After establishing the basic cases, we can use Lemma~\ref{basic_lem_2} and Lemma~\ref{basic_lem_3} to realize the inductive proof by recursively deleting four leaves. However, since not all phylogenetic trees have inductive leaf deletion structures (particularly, Y-like trees hinder this operation), it is necessary to handle such special structures separately. After excluding these special cases, for all remaining phylogenetic trees, we can ensure the existence of four deletable leaves, thereby constructing a concise recursive induction path. This approach not only preserves the generality of the inductive proof but also ensures the completeness of the argument through preprocessing of special structures.

\subsection{Phylogenetic Trees with partially 15 Leaves}
\label{sec:15-leaf}

This section analyzes a special class of phylogenetic trees $\mathcal{T}$ with 15 leaves, which must satisfy the following two constraints: (1) no internal subtree isomorphic to the snowflake structure; (2) no 5 internal subtrees isomorphic to the 3-star.

We first classify the trees based on the number of interior edges in the longest path, and then through constructive proofs, we find that there exist three sets $\mathcal{C}$ of 5-state characters that can uniquely determine the topological structure of the tree $\mathcal{T}$. This conclusion lays an important foundation for subsequent research, and the next section will expand on this to characterize phylogenetic trees with 16 to 19 leaves.

\begin{Lem}\label{pre_n=15_0}
Let $\mathcal{T}$ be a phylogenetic $X$-tree with $X = \{1, 2, \ldots, 15\}$. If $\mathcal{T}$ satisfies conditions:
\begin{enumerate}
  \item[(i)] $\mathcal{T}$ has no internal subtree isomorphic to the snowflake,
  \item[(ii)] $\mathcal{T}$ has no 5 internal subtrees isomorphic to the 3-star,
\end{enumerate}
then the number of interior edges in the longest path of $\mathcal{T}$ is at least 6.
\end{Lem}
\begin{proof}
By attempting to construct phylogenetic trees, it is easy to see that there is no phylogenetic tree with 15 leaves where the longest path has 4 interior edges. Additionally, it can be easily observed that, up to isomorphism, the phylogenetic $X$-tree shown in Fig.~\ref{fig:internal_path_5} is the only one with 15 leaves where the longest path has 5 interior edges. Let $\mathcal{T}'$ denote the phylogenetic $X$-tree shown in Fig.~\ref{fig:internal_path_5}. Observing $\mathcal{T}'$, we find that it has an internal subtree isomorphic to the snowflake, as indicated by the dashed lines in Fig.~\ref{fig:internal_path_5}. Therefore, the number of interior edges in the longest path of a valid phylogenetic $X$-tree satisfying the conditions is at least 6.
\end{proof}

For a specific phylogenetic tree, to prove the existence of a set $\mathcal{C}$ of three 5-state characters that defines $\mathcal{T}$, we use a method of assigning interior edges to three sets, forming a character set $\mathcal{C}$ with three characters displayed by these sets. We then draw the corresponding $\text{Int}(\mathcal{C})$ and prove the conclusion using Theorem~\ref{basic_thm_judgment}. Since the labeling order of leaves does not affect the assignment of interior edges, we only need to consider the structure of the phylogenetic tree, which will not be repeated in the following lemma proofs.

Let $F$ be a subset of edges of a phylogenetic $X$-tree $\mathcal{T}$, and let $C_1, C_2, \ldots, C_k$ denote the connected components of $\mathcal{T} \setminus F$. For all $i$, let $E_i$ denote the set of interior edges of $\mathcal{T}$ in $C_i$. Note that for some $i$, the set $E_i$ may be empty. We say that $F$ \emph{separates} the interior edges of $\mathcal{T}$ not in $F$ into the sets $E_1, E_2, \ldots, E_k$.

\begin{Lem}\label{pre_n=15_6}
Let $\mathcal{T}$ be a phylogenetic $X$-tree with $|X| = 15$. If $\mathcal{T}$ satisfies conditions~(i),(ii) and the number of interior edges in its longest path is 6, then there exists a set $\mathcal{C}$ of three 5-state characters that defines $\mathcal{T}$.
\end{Lem}
\begin{proof}
Let $\mathring{P}$ be the set of interior edges in the longest path $P$ of $\mathcal{T}$, so $|\mathring{P}| = 6$. $\mathring{P}$ separates the interior edges of $\mathcal{T}$ not in $\mathring{P}$ into $E_1, E_2, \ldots, E_7$, where $|E_1| = |E_7| = 0$. Since $\mathcal{T}$ has 12 interior edges, we only need to consider the distribution of the remaining 6 interior edges in $E_2, E_3, \ldots, E_6$. Due to the longest path having 6 interior edges, we observe that $|E_2|, |E_6| \leq 1$, $|E_3|, |E_5| \leq 3$, and $|E_4| \leq 7$. Let $F_1, F_2, F_3$ be a partition of the interior edges of $\mathcal{T}$ such that $|F_1| = |F_2| = |F_3| = 4$. For $i \in \{1, 2, 3\}$, let $\chi_i$ be displayed by $F_i$. Let $A$ represent the distribution of the remaining interior edges (excluding $\mathring{P}$), where $A = m_1 + m_2 + \ldots + m_k$ denotes distributing the remaining 6 interior edges into $k$ sets $E_2, E_3, \ldots, E_6$ with sizes $m_1, m_2, \ldots, m_k$.

(1) When $A = 2 + 1 + 1 + 1 + 1$, $\mathcal{T}$ has 5 internal subtrees isomorphic to the 3-star, which violates condition~(ii).

(2) When $A = 3 + 1 + 1 + 1$, up to isomorphism, there are two types of phylogenetic tree structures satisfying conditions (i) and (ii), as shown in Fig.~\ref{fig:P=6_3+1+1+1}.
\begin{figure}[h]
  \centering
  \subfigure[]
  {
   \begin{minipage}[b]{.45\linewidth}
     \centering
     \includegraphics[width=0.85\textwidth]{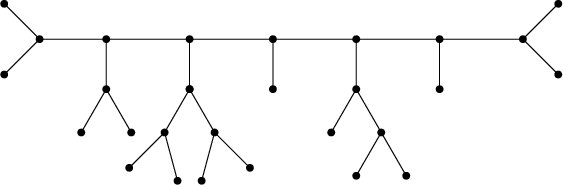}
   \end{minipage}
  }
  \subfigure[]
  {
   \begin{minipage}[b]{.45\linewidth}
     \centering
     \includegraphics[width=0.85\textwidth]{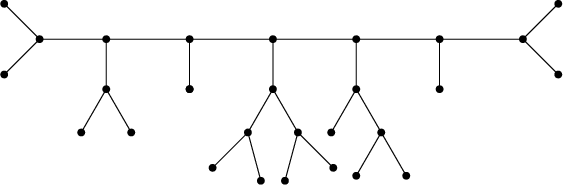}
   \end{minipage}
  }
  \caption{Phylogenetic tree structures for $A = 3 + 1 + 1 + 1$.}
  \label{fig:P=6_3+1+1+1}
\end{figure}
From left to right, the 1st, 3rd, and 5th interior edges on the longest path $P$ are assigned to $F_1$, and the 2nd, 4th, and 6th interior edges are assigned to $F_2$. For $i \in \{2, 3, 4, 5, 6\}$, the interior edges in $E_i$ connected to $P$ are assigned to $F_3$, and the remaining two interior edges are assigned to $F_1$ and $F_2$, respectively. $\chi_i$ is displayed by $F_i$. By drawing $\text{Int}(\{\chi_1, \chi_2, \chi_3\})$ accordingly, we conclude that $\{\chi_1, \chi_2, \chi_3\}$ defines $\mathcal{T}$ according to Theorem~\ref{basic_thm_judgment}.

(3) When $A = 2 + 2 + 1 + 1$, similar to case (2), interior edges are assigned to $F_1$, $F_2$, and $F_3$. By Theorem~\ref{basic_thm_judgment}, $\{\chi_1, \chi_2, \chi_3\}$ defines $\mathcal{T}$.

(4) When $A = 4 + 1 + 1$, $E_4 = 4$. Since $|\mathring{P}| = 6$, the 4 interior edges in $E_4$ are not on a single path. Up to isomorphism, there are three types of phylogenetic tree structures satisfying conditions (i) and (ii), as shown in Fig.~\ref{fig:P=6_4+1+1}.
\begin{figure}
  \centering
  \subfigure[]
  {
   \begin{minipage}[b]{.45\linewidth}
     \centering
     \includegraphics[width=0.90\textwidth]{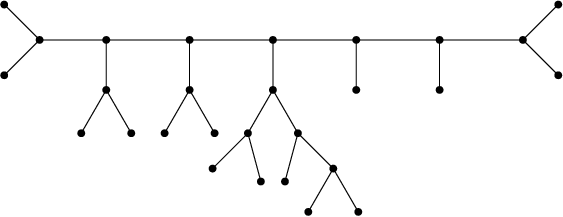}
   \end{minipage}
  }
  \subfigure[]
  {
   \begin{minipage}[b]{.45\linewidth}
     \centering
     \includegraphics[width=0.9\textwidth]{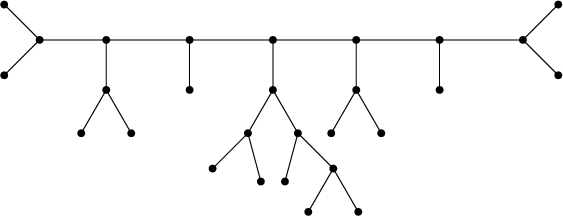}
   \end{minipage}
  }
  \subfigure[]
  {
   \begin{minipage}[b]{.45\linewidth}
     \centering
     \includegraphics[width=0.9\textwidth]{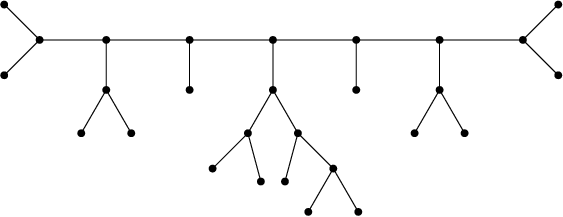}
   \end{minipage}
  }
  \caption{Phylogenetic tree structures for $A = 4 + 1 + 1$.}
  \label{fig:P=6_4+1+1}
\end{figure}
In case (a), a path can be found such that the longest path has 6 interior edges, with the remaining edges assigned as $A = 3 + 1 + 1 + 1$. For cases (b) and (c), from left to right, the 1st, 3rd, and 5th interior edges on $P$ are assigned to $F_1$, and the 2nd, 4th, and 6th interior edges are assigned to $F_2$. For $i \in \{2, 3, 4, 5, 6\}$, the interior edges in $E_i$ connected to $P$ and the interior edge farthest from $P$ in $E_4$ are assigned to $F_3$. The remaining two interior edges are assigned to $F_1$ and $F_2$, respectively. $\chi_i$ is displayed by $F_i$, and $\text{Int}(\{\chi_1, \chi_2, \chi_3\})$ confirms that $\{\chi_1, \chi_2, \chi_3\}$ defines $\mathcal{T}$ by Theorem~\ref{basic_thm_judgment}.

(5) When $A = 3 + 2 + 1$, up to isomorphism, besides the three phylogenetic trees in Fig.~\ref{fig:P=6_3+2+1}, the remaining cases can be handled by similar edge assignments or by selecting another longest path to reduce to previously discussed cases.
\begin{figure}
  \centering
  \subfigure[]
  {
   \begin{minipage}[b]{.45\linewidth}
     \centering
     \includegraphics[width=0.9\textwidth]{fig.P=6_3+2+1_1.eps}
   \end{minipage}
  }
  \subfigure[]
  {
   \begin{minipage}[b]{.45\linewidth}
     \centering
     \includegraphics[width=0.9\textwidth]{fig.P=6_3+2+1_2.eps}
   \end{minipage}
  }
  \subfigure[]
  {
   \begin{minipage}[b]{.5\linewidth}
     \centering
     \includegraphics[width=0.9\textwidth]{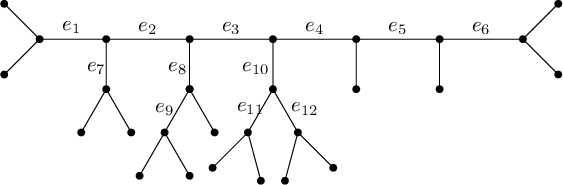}
   \end{minipage}
  }
  \caption{Phylogenetic tree structures for $A = 3 + 2 + 1$.}
  \label{fig:P=6_3+2+1}
\end{figure}
For cases (a) and (b) in Fig.~\ref{fig:P=6_3+2+1}, the 1st, 3rd, and 5th interior edges on $P$ are assigned to $F_1$, the 2nd and 4th interior edges to $F_2$, and the 6th interior edge to $F_3$. For $i \in \{2, 3, 4, 5, 6\}$, the interior edges in $E_i$ connected to $P$ and the interior edge farthest from the longest path in the size-2 $E_i$ are assigned to $F_2$. The remaining two interior edges are assigned to $F_1$ and $F_2$, respectively. $\chi_i$ is displayed by $F_i$, and $\text{Int}(\{\chi_1, \chi_2, \chi_3\})$ confirms the definition by Theorem~\ref{basic_thm_judgment}. For case (c), interior edges $e_1, e_3, e_5, e_{11}$ are assigned to $F_1$, $e_7, e_8, e_4, e_{12}$ to $F_2$, and $e_2, e_6, e_9, e_{10}$ to $F_3$. $\chi_i$ is displayed by $F_i$, and $\text{Int}(\{\chi_1, \chi_2, \chi_3\})$ confirms the definition.

(6) When $A = 2 + 2 + 2$, $E_3 = E_4 = E_5 = 2$. From left to right, the 1st, 3rd, and 5th interior edges on $P$ are assigned to $F_1$, the 2nd and 4th to $F_2$, and the 6th to $F_3$. The interior edges in $E_3, E_4, E_5$ connected to $P$ are assigned to $F_3$, with the remaining three interior edges assigned as one to $F_1$ and two to $F_2$. $\chi_i$ is displayed by $F_i$, and $\text{Int}(\{\chi_1, \chi_2, \chi_3\})$ confirms the definition.

(7) When $A = 3 + 3$, $A = 4 + 2$, $A = 5 + 1$, or $A = 6$, up to isomorphism, phylogenetic trees satisfying conditions (i) and (ii) can be reduced to previously discussed cases by selecting another longest path.

In summary, when $\mathcal{T}$ satisfies conditions (i), (ii) and the longest path has 6 interior edges, there exists a set $\mathcal{C}$ of three 5-state characters that defines $\mathcal{T}$.
\end{proof}

\begin{Lem}\label{pre_n=15_7}
Let $\mathcal{T}$ be a binary phylogenetic $X$-tree with $|X| = 15$. If $\mathcal{T}$ satisfies conditions (i), (ii) and the number of interior edges in its longest path is 7, then there exists a set $\mathcal{C}$ of three 5-state characters that defines $\mathcal{T}$.
\end{Lem}
\begin{proof}
Let $\mathring{P}$ be the set of interior edges in the longest path $P$ of $\mathcal{T}$, so $|\mathring{P}| = 7$. $\mathring{P}$ separates the interior edges of $\mathcal{T}$ not in $\mathring{P}$ into $E_1, E_2, \ldots, E_8$, where $|E_1| = |E_8| = 0$. Since $\mathcal{T}$ has 12 interior edges, we only need to consider the distribution of the remaining 5 interior edges in $E_2, E_3, \ldots, E_7$. Due to the longest path having 7 interior edges, we observe that $|E_2|, |E_7| \leq 1$, $|E_3|, |E_6| \leq 3$, and $|E_4|, |E_5| \leq 7$. Let $F_1, F_2, F_3$ be a partition of the interior edges of $\mathcal{T}$ such that $|F_1| = |F_2| = |F_3| = 4$. For $i \in \{1, 2, 3\}$, let $\chi_i$ be displayed by $F_i$. Let $A$ represent the distribution of the remaining interior edges (excluding $\mathring{P}$), where $A = m_1 + m_2 + \ldots + m_k$ denotes distributing the remaining 5 interior edges into $k$ sets $E_2, E_3, \ldots, E_7$ with sizes $m_1, m_2, \ldots, m_k$.

(1)	When $A = 1 + 1 + 1 + 1 + 1$, $\mathcal{T}$ has 5 internal subtrees isomorphic to the 3-star, violating condition~(ii).
(2) When $A = 2 + 1 + 1 + 1$, following the construction in Lemma \ref{pre_n=15_6}(2), we construct three sets $\mathcal{C}$ of 5-state characters for $\mathcal{T}$.  

For example, in Fig.~\ref{fig:P=7_2+1+1+1}, from left to right, the 1st, 3rd, 5th, and 7th interior edges on the longest path $P$ are assigned to $F_1$, and the 2nd, 4th, and 6th interior edges are assigned to $F_2$. For $i \in \{2,3,4,5,6,7\}$, the interior edges in $E_i$ connected to $P$ are assigned to $F_3$, and the remaining interior edge is assigned to $F_2$. The character $\chi_i$ is represented by $F_i$. By drawing $\text{Int}(\{\chi_1, \chi_2, \chi_3\})$ accordingly, $\{\chi_1, \chi_2, \chi_3\}$ defines $\mathcal{T}$ by Theorem \ref{basic_thm_judgment}. Other cases are constructed similarly.  

\begin{figure}[H]
\centering
\includegraphics[width=0.5\linewidth]{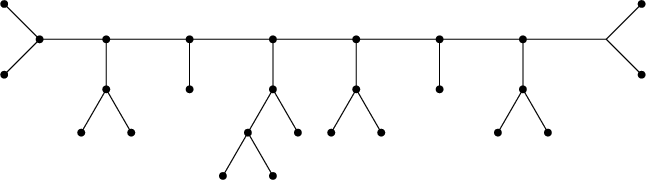}
\caption{Phylogenetic tree structure for $A = 2 + 1 + 1 + 1$.}
\label{fig:P=7_2+1+1+1}
\end{figure}

(3) When $A = 3 + 1 + 1$, assume $|E_i| = 3$ for $i \in \{3,4,5,6\}$. By symmetry, we only consider $i \in \{3,4\}\}$. Two cases arise based on the connectivity of edges in $E_i$:  

\textbf{Case 1:} The three interior edges in $E_i$ form a path. Since $P$ is the longest path, we have $i = 4$. As shown in Fig.~\ref{fig:P=7_3+1+1_1}, the 1st, 3rd, 5th, and 7th interior edges on $P$ are assigned to $F_1$, and the 2nd, 4th, and 6th interior edges are assigned to $F_2$. For $j \in \{2,3,4,5,6,7\}$, the interior edges in $E_j$ connected to $P$ are assigned to $F_3$, and the interior edge in $E_i$ farthest from $P$ is also assigned to $F_3$, with the remaining edge assigned to $F_2$. The characters $\chi_i$ defined by $F_i$ form $\text{Int}(\{\chi_1, \chi_2, \chi_3\})$, which defines $\mathcal{T}$ by Theorem \ref{basic_thm_judgment}.  

\begin{figure}[H]
\centering
\includegraphics[width=0.5\linewidth]{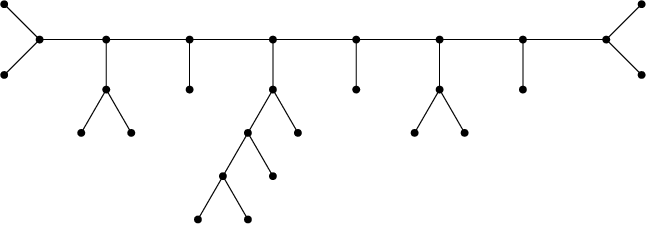}
\caption{Phylogenetic tree structure when three interior edges in $E_i$ form a path.}
\label{fig:P=7_3+1+1_1}
\end{figure}

\textbf{Case 2:} The three interior edges in $E_i$ share a common end (i.e., do not form a path). Let $|E_i| = 3$, $|E_j| = |E_k| = 1$, and $j < k$. The phylogenetic tree satisfying conditions (i) and (ii) has no consecutive indices $j, i, k$.  

\textbf{Subcase a:} If $j = 2$, $k = 7$, there are two isomorphic phylogenetic trees as shown in Fig.~\ref{fig:P=7_3+1+1_a}. The interior edges in $E_i, E_j, E_k$ connected to $P$ and the 5th interior edge of $P$ are assigned to $F_3$. The remaining edges are distributed to $F_1$ and $F_2$ such that adjacent edges belong to different sets. The characters $\chi_i$ define $\mathcal{T}$ as before.  

\begin{figure}[H]
\centering
\subfigure[]
{
\begin{minipage}[b]{0.45\linewidth}
\centering
\includegraphics[width=0.9\textwidth]{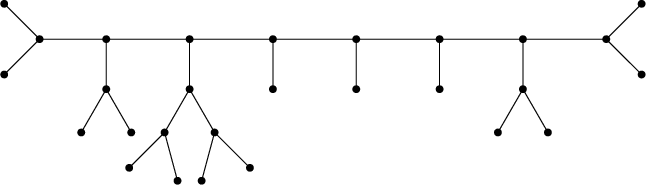}
\end{minipage}
}
\subfigure[]
{
\begin{minipage}[b]{0.45\linewidth}
\centering
\includegraphics[width=0.9\textwidth]{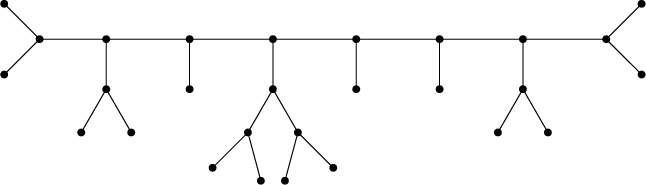}
\end{minipage}
}
\caption{Phylogenetic tree structures for $|E_2| = |E_7| = 1$.}
\label{fig:P=7_3+1+1_a}
\end{figure}

\textbf{Subcase b:} If $i, j, k$ or $j, k, i$ are consecutive indices, there are three isomorphic phylogenetic trees as shown in Fig.~\ref{fig:P=7_3+1+1_b}. For figures (a) and (c), the 1st and 7th interior edges of $P$, the interior edges in $E_i$ connected to $P$, and the interior edges adjacent to $E_j, E_k$ on $P$ are assigned to $F_3$. For figure (b), the 5th and 7th interior edges of $P$ are assigned to $F_3$ similarly. The remaining four edges of $P$ are split evenly between $F_1$ and $F_2$, with $E_i$'s remaining edges and $E_j, E_k$'s edges assigned to $F_1, F_2$ respectively.  

\begin{figure}[H]
\centering
\subfigure[]
{
\begin{minipage}[b]{0.45\linewidth}
\centering
\includegraphics[width=0.9\textwidth]{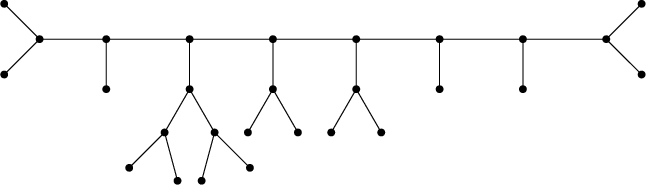}
\end{minipage}
}
\subfigure[]
{
\begin{minipage}[b]{0.45\linewidth}
\centering
\includegraphics[width=0.9\textwidth]{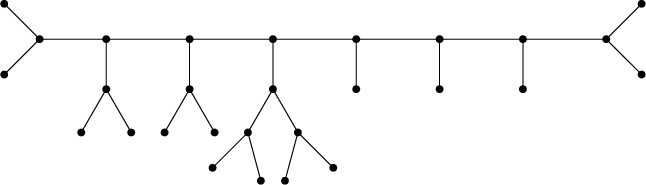}
\end{minipage}
}
\subfigure[]
{
\begin{minipage}[b]{0.5\linewidth}
\centering
\includegraphics[width=0.9\textwidth]{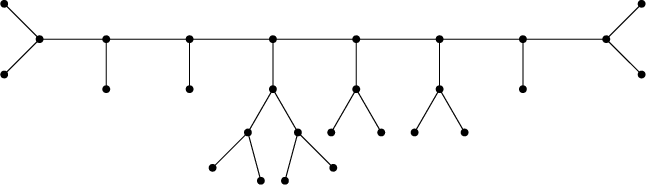}
\end{minipage}
}
\caption{Phylogenetic tree structures when $i, j, k$ are consecutive.}
\label{fig:P=7_3+1+1_b}
\end{figure}

\textbf{Subcase c:} Excluding cases a and b, one of the 1st or 7th edges of $P$ (not both) is assigned to $F_3$ to ensure non-adjacent edges in $F_3$. The remaining edges of $P$ and $E_i$ are split evenly between $F_1$ and $F_2$.  

(4) When $A = 2 + 2 + 1$, if $|E_i| = 2$, then $i \neq 2, 7$, and at most one of $|E_2|, |E_7|$ is non-zero. The interior edges in $E_i$ connected to $P$ and one of the 1st or 7th edges of $P$ (non-adjacent) are assigned to $F_3$. The remaining six edges of $P$ and two interior edges are split evenly between $F_1$ and $F_2$.  

(5) When $A = 4 + 1$, construct $\mathcal{C}$ analogously to (4).  

(6) When $A = 3 + 2$, since $|\mathring{P}| = 7$, $|E_2| = |E_7| = 0$. Let $|E_i| = 3$, $|E_j| = 2$. The interior edges in $E_i, E_j$ connected to $P$, and the 1st and 7th edges of $P$ are assigned to $F_3$. The remaining edges of $E_i, E_j$ and $P$ are distributed to $F_1, F_2$ with non-adjacent constraints.  

(7) When $A = 5$, there is a unique isomorphic phylogenetic tree. Let $|E_4| = 5$. The interior edge of $E_4$ connected to $P$, two distant interior edges of $E_4$, and the 7th edge of $P$ are assigned to $F_3$. The remaining edges of $P$ and $E_4$ are split evenly between $F_1$ and $F_2$.  
\end{proof}

\begin{Lem}\label{pre_n=15_8}
Let $\mathcal{T}$ be a phylogenetic $X$-tree with $|X| = 15$. If $\mathcal{T}$ satisfies conditions (i), (ii) and the longest path has 8 interior edges, then there exists a set $\mathcal{C}$ of three 5-state characters that defines $\mathcal{T}$.  
\end{Lem}

\begin{proof}
Follow the construction in Lemmas \ref{pre_n=15_6} and \ref{pre_n=15_7}.  

(1) When $A = 1 + 1 + 1 + 1$, the 1st, 3rd, 5th, 7th interior edges of $P$ are assigned to $F_1$, and the 2nd, 4th, 6th, 8th edges to $F_2$. For $i \in \{2,3,4,5,6,7,8\}$, interior edges in $E_i$ connected to $P$ are assigned to $F_3$.  

(2) When $A = 2 + 1 + 1$, if $|E_2| = |E_8| = 1$ and $|E_i| = 2$ ($i \in \{3,4,5\}$), assign the 1st, 3rd, 5th, 7th edges of $P$ to $F_1$, the 2nd, 4th, 8th edges to $F_2$, and distant edges of $E_i$ to $F_2$. Connecting edges of $E_2, E_8, E_i$ and the 6th edge of $P$ are assigned to $F_3$. For other cases, adjust edge assignments similarly.  

(3)-(5) Follow analogous constructions for $A = 3 + 1$, $A = 2 + 2$, and $A = 4$, ensuring non-adjacent edges in $F_3$ and balanced distribution in $F_1, F_2$.  

In all cases, $\chi_i$ defined by $F_i$ forms $\text{Int}(\{\chi_1, \chi_2, \chi_3\})$, which defines $\mathcal{T}$ by Theorem \ref{basic_thm_judgment}.  
\end{proof}

\begin{Lem}\label{pre_n=15_9}
Let $\mathcal{T}$ be a binary phylogenetic $X$-tree with $|X| = 15$. If $\mathcal{T}$ satisfies conditions (i), (ii) and the longest path has at least 9 interior edges, then there exists a set $\mathcal{C}$ of three 5-state characters that defines $\mathcal{T}$.  
\end{Lem}

\begin{proof}
Adapt the constructions from Lemmas~\ref{pre_n=15_6}, \ref{pre_n=15_7}, and \ref{pre_n=15_8} to define $F_1, F_2, F_3$. The characters $\chi_i$ from $F_i$ form $\text{Int}(\{\chi_1, \chi_2, \chi_3\})$, which defines $\mathcal{T}$ by Theorem \ref{basic_thm_judgment}.  
\end{proof}

\begin{Cor}\label{pre_n=15_all}
Let $\mathcal{T}$ be a phylogenetic $X$-tree with $|X| = 15$. If $\mathcal{T}$ satisfies conditions~(i) and~(ii) as defined above,
then there exists a set $\mathcal{C}$ of three 5-state characters that defines $\mathcal{T}$.  
\end{Cor}


\subsection{\texorpdfstring{Phylogenetic Trees with 16 $\sim$ 19 Leaves}{Phylogenetic Trees with 16-19 Leaves}}\label{sec:16-19-leaf}

This section systematically analyzes the base cases (topologies with 16 $\sim$ 19 leaves) in the inductive proof for phylogenetic trees. First, we take the phylogenetic tree with 19 leaves as the research object: except for two special topologies, any phylogenetic tree $\mathcal{T}$ with 19 leaves can find a phylogenetic tree with 15 leaves contained in it as its restriction. Based on the theoretical foundation in Section~\ref{sec:CDT-cy} and combined with the proof of the existence of the character set in Section~\ref{sec:15-leaf}, it can be deduced that there exists a set $\mathcal{C}$ of four 5-state characters such that $\mathcal{C}$ defines $\mathcal{T}$. For the base cases with 16 - 18 leaves, this section adopts a leaf extension strategy --- by gradually adding leaves to expand the original tree structure into a phylogenetic tree with 19 leaves, and then applying the aforementioned conclusions to complete the proof.

\begin{figure}[H]
\centering
\subfigure[$\mathcal{T}_{1}$]
{
\begin{minipage}[b]{0.45\linewidth}
\centering
\includegraphics[width=0.90\textwidth]{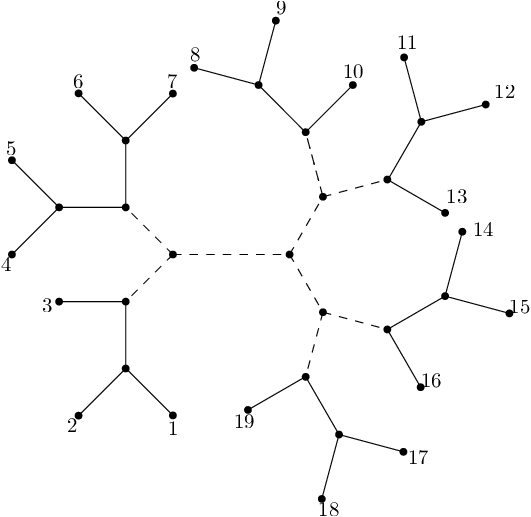}
\end{minipage}
}
\subfigure[$\mathcal{T}_{2}$]
{
\begin{minipage}[b]{0.45\linewidth}
\centering
\includegraphics[width=0.90\textwidth]{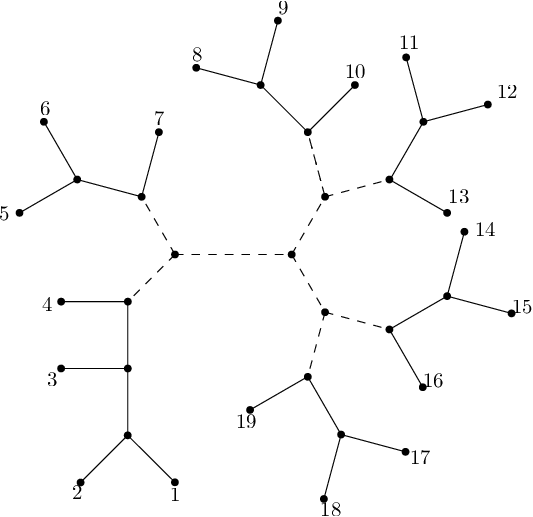}
\end{minipage}
}
\caption{Two phylogenetic trees on $\{1,2,\ldots,19\}$.}
\label{fig:n=19_special}
\end{figure}

\begin{Lem}\label{pre_n=19_0}
Let $\mathcal{T}_{1}$ and $\mathcal{T}_{2}$ be phylogenetic $X$-trees, where $X = \{1,2,\ldots,19\}$, and their structures are shown in Fig.\ref{fig:n=19_special}. For any $s \in \{1,2\}$, there exists a set $\mathcal{C}_{s}$ of four 5-state characters such that $\mathcal{C}_{s}$ defines $\mathcal{T}_{s}$.
\end{Lem}
\begin{proof}
Let $\mathcal{C}_{1} = \{\chi_{11},\chi_{12},\chi_{13},\chi_{14}\}$ be a set of characters on $\{1,2,\ldots,19\}$, where
\[
\begin{aligned}
\pi(\chi_{11}) &= \{\{1,2\},\{4,5\},\{8,9,10\},\{3,6,7\}, \{11,12,13,14,15,16,17,18,19\}\}, \\
\pi(\chi_{12}) &= \{\{1,2,3\},\{6,7\},\{11,12,13\},\{4,5,8,9,10\}, \{14,15,16,17,18,19\}\}, \\
\pi(\chi_{13}) &= \{\{1,2,3,19\},\{4,5,6,7\},\{8,9,10,11,12,13\}, \{14,15,16\},\{17,18\}\}, \\
\pi(\chi_{14}) &= \{\{8,9\},\{11,12\},\{14,15\},\{17,18,19\},\{1,2,3,4,5,6,7,10,13,16\}\}.
\end{aligned}
\]

Let $\mathcal{C}_{2} = \{\chi_{21},\chi_{22},\chi_{23},\chi_{24}\}$ be a set of characters on $\{1,2,\ldots,19\}$, where
\[
\begin{aligned}
\pi(\chi_{21}) &= \{\{1,2,3\},\{5,6\},\{4,7\},\{11,12\},\{8,9,10,13,14,15,16,17,18,19\}\}, \\
\pi(\chi_{22}) &= \{\{1,2,3,4\},\{8,9,10\},\{5,6,7,11,12,13\},\{14,15\},\{16,17,18,19\}\}, \\
\pi(\chi_{23}) &= \{\{1,2,3,4,19\},\{5,6,7\},\{8,9,10,11,12,13\}, \{14,15,16\},\{17,18\}\}, \\
\pi(\chi_{24}) &= \{\{1,2\},\{8,9\},\{11,12,13\},\{17,18,19\},\{3,4,5,6,7,10,14,15,16\}\}.
\end{aligned}
\]
For all $s \in \{1,2\}$, it is easy to verify that $\mathcal{C}_{s}$ is convex on $\mathcal{T}_{s}$ and that $\mathcal{T}_{s}$ is distinguished by $\mathcal{C}_{s}$.

For $\mathcal{C}_{2}$, draw $\text{Int}(\mathcal{C}_{2})$, remove the vertices whose induced subgraph of neighbors is a complete graph, and then remove the vertices adjacent to all the remaining vertices. Only by connecting $(\chi_{23}, \{1,2,3,4,19\})$ and $(\chi_{22}, \{5,6,7,11,12,13\})$ can we obtain the unique restricted completely chordal graph of $\text{Int}(\mathcal{C}_{2})$. For $\mathcal{C}_{1}$, draw $\text{Int}(\mathcal{C}_{1})$, remove the vertices with a complete neighbor subgraph, and then remove the vertices adjacent to all the remaining vertices. The remaining graph is a cycle composed of six vertices, which are $(\chi_{12}, \{4,5,8,9,10\})$, $(\chi_{13}, \{8,9,10,11,12,13\})$, $(\chi_{11}, \{11,12,13,14,15,16,17,18,19\})$, $(\chi_{13}, \{1,2,3,19\})$, $(\chi_{11}, \{3,6,7\})$, and $(\chi_{13}, \{4,5,6,7\})$ in sequence. Adding three edges, whose ends are respectively $(\chi_{12}, \{4,5,8,9,10\})$ and $(\chi_{11}, \{11,12,13,14,15,16,17,18,19\})$, $(\chi_{12}, \{4,5,8,9,10\})$ and $(\chi_{13}, \{1,2,3,19\})$, and $(\chi_{12}, \{4,5,8,9,10\})$ and $(\chi_{11}, \{3,6,7\})$, in $\text{Int}(\mathcal{C}_{1})$ can obtain the unique restricted completely chordal graph of $\text{Int}(\mathcal{C}_{1})$. By Theorem \ref{basic_thm_judgment}, for any $s \in \{1,2\}$, $\mathcal{C}_{s}$ defines $\mathcal{T}_{s}$.
\end{proof}

Lemma \ref{pre_n=19_0} demonstrates two phylogenetic trees with 19 leaves of special structures.

\begin{Cor}\label{pre_n=19_all}
Let $\mathcal{T}$ be a phylogenetic $X$-tree with 19 leaves. Then there exists a set $\mathcal{C}$ of four 5-state characters such that $\mathcal{C}$ defines $\mathcal{T}$.
\end{Cor}
\begin{proof}
Let $\mathcal{T}$ be a phylogenetic $X$-tree with 19 leaves. If we can find a set $L$ of 4 pendant edges such that the interior edges associated with them are non-adjacent, and the phylogenetic tree $\widetilde{\mathcal{T}}$ obtained by suppressing the degree-two vertices in $\mathcal{T} \setminus L$ is a restriction of $\mathcal{T}$ and satisfies conditions~(i) and (ii). By Corollary \ref{pre_n=15_all}, there exists a set $\widetilde{\mathcal{C}}$ of three 5-state characters such that $\widetilde{\mathcal{C}}$ defines $\widetilde{\mathcal{T}}$. Let $\widetilde{F} \subseteq E(\widetilde{\mathcal{T}})$ be a $\widetilde{\mathcal{T}}$-representable subset. According to Lemmas~\ref{basic_lem_2} and~\ref{basic_lem_3}, $\widetilde{\mathcal{C}}_{\widetilde{F}}$ (where $\widetilde{F}$ is a $\widetilde{\mathcal{T}}$-representable subset) together with the characters $\chi$ displayed by the interior edges associated with $L$ can define $\mathcal{T}$. At this time, the case of $n = 19$ is reduced to the case of $n = 15$. Now we show that, except for the phylogenetic trees isomorphic to $\mathcal{T}_{1}$ and $\mathcal{T}_{2}$ in Lemma \ref{fig:n=19_special}, for any $\mathcal{T}$, there exists a phylogenetic $\widetilde{X}$-tree $\widetilde{\mathcal{T}}$ with $|\widetilde{X}| = 15$, and $\widetilde{\mathcal{T}}$ is a restriction of $\mathcal{T}$ and satisfies conditions~(i) and (ii).

For condition~(i), we conduct a classification discussion based on the number of internal subtrees isomorphic to snowflakes in $\mathcal{T}$. If $\mathcal{T}$ has only one internal subtree isomorphic to the snowflake, denote the pendant edges of this internal subtree as $e_{i}$, where $i \in \{1,2,3,4,5,6\}$, and $E_{0} = \{e_{1},\ldots,e_{6}\}$. If each $e_{i}$ has at least one adjacent interior edge in $\mathcal{T}$ that is not in $E_{0}$, then $\mathcal{T}$ is isomorphic to $\mathcal{T}_{1}$ or $\mathcal{T}_{2}$ in Lemma \ref{fig:n=19_special}. Thus, in the remaining phylogenetic trees, there exists some $e_{k}$ whose adjacent edges in $\mathcal{T} \setminus E_{0}$ are not interior edges, that is, the two adjacent edges of $e_{k}$ in $\mathcal{T} \setminus E_{0}$ are both pendant edges in $\mathcal{T}$. Select one of them and put it into $L$, and then select 3 pendant edges in $\mathcal{T}$ and put them into $L$ so that the interior edges associated with $L$ are non-adjacent. Then the phylogenetic tree $\widetilde{\mathcal{T}}$ obtained by suppressing the degree-two vertices in $\mathcal{T} \setminus L$ does not contain internal subtrees isomorphic to snowflakes.

\begin{figure}[H]
\centering
\subfigure[]
{
\begin{minipage}[b]{0.45\linewidth}
\centering
\includegraphics[width=0.60\textwidth]{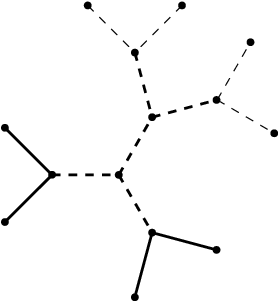}
\end{minipage}
}
\subfigure[]
{
\begin{minipage}[b]{0.45\linewidth}
\centering
\includegraphics[width=0.70\textwidth]{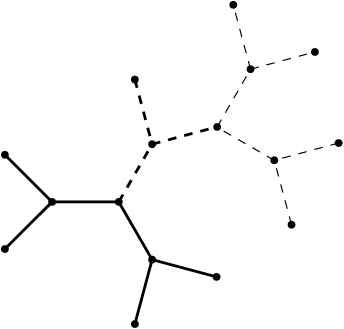}
\end{minipage}
}
\caption{Structures with two internal subtrees isomorphic to the snowflake.}
\label{fig:n=19_two-snowflake}
\end{figure}

If $\mathcal{T}$ has two internal subtrees isomorphic to snowflakes, due to $n = 19$, these two internal subtrees have the same edges, and their structures are only the two types shown in Fig.~\ref{fig:n=19_two-snowflake}. In Fig.~\ref{fig:n=19_two-snowflake}, the thick lines and the dashed lines represent two internal subtrees isomorphic to snowflakes, and the set of pendant edges of the internal subtrees is denoted as $E_{0}$. In (a), each pendant edge in $\mathcal{T}$ extends to at least two leaves. Since $n = 19$, there is at least one interior edge $e_{i}$ of the thick lines in $E_{0}$ that extends to only two leaves, that is, the two adjacent edges of $e_{i}$ in $\mathcal{T} \setminus E_{0}$ are both pendant edges in $\mathcal{T}$. Similarly, there is at least one interior edge $e_{j}$ of the dashed lines in $E_{0}$ that extends to only two leaves, that is, the two adjacent edges of $e_{j}$ in $\mathcal{T} \setminus E_{0}$ are both pendant edges in $\mathcal{T}$. Put $e_{i}$ and $e_{j}$ into $L$, and then select two pendant edges in $\mathcal{T}$ and put them into $L$ so that the interior edges associated with $F$ are all non-adjacent. Then the phylogenetic tree $\widetilde{\mathcal{T}}$ obtained by suppressing the degree-two vertices in $\mathcal{T} \setminus L$ does not contain internal subtrees isomorphic to snowflakes. The same consideration applies to (b), and we can find $L$ such that the phylogenetic tree $\widetilde{\mathcal{T}}$ obtained by suppressing the degree-two vertices in $\mathcal{T} \setminus L$ does not contain internal subtrees isomorphic to snowflakes. If $\mathcal{T}$ has three internal subtrees isomorphic to snowflakes, then $\mathcal{T}$ requires at least 20 leaves, which does not conform to $n = 19$.

\begin{figure}[H]
\centering
\subfigure[]
{
\begin{minipage}[b]{0.45\linewidth}
\centering
\includegraphics[width=0.60\textwidth]{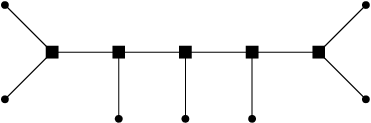}
\end{minipage}
}
\subfigure[]
{
\begin{minipage}[b]{0.45\linewidth}
\centering
\includegraphics[width=0.45\textwidth]{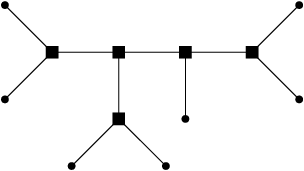}
\end{minipage}
}
\caption{Structures with five internal subtrees isomorphic to the 3-star.}
\label{fig:n=19_five_3-star}
\end{figure}

For condition~(ii), denote the set of degree-three vertices of the internal subtrees isomorphic to the 3-star as $W = \{w_{1},w_{2},w_{3},w_{4},w_{5}\}$. Consider the distribution of the vertices in $W$ in $\mathcal{T}$. If the vertices in $W$ are on a straight line, then the internal subtrees containing isomorphic the 3-star have at least 7 pendant edges, as shown in (a) of Fig.~\ref{fig:n=19_five_3-star}. If four vertices in $W$ are on a straight line, then the internal subtrees containing isomorphic the 3-star also have at least 7 pendant edges, as shown in (b) of Fig.~\ref{fig:n=19_five_3-star}. Denote the set of pendant edges of the internal subtrees as $E_{0}$. Only if each edge in $E_{0}$ extends to at least three leaves in $\mathcal{T}$, can the phylogenetic tree $\tilde{\mathcal{T}}$ obtained by suppressing the degree-two vertices in $\mathcal{T} \setminus L$ contain five internal subtrees isomorphic to the 3-star no matter how $L$ is taken. However, if each edge in $E_{0}$ extends to at least three leaves in $\mathcal{T}$, $\mathcal{T}$ requires at least 21 leaves, which does not conform to $n = 19$. Therefore, there does not exist a phylogenetic tree $\mathcal{T}$ with 19 leaves such that the phylogenetic tree obtained by removing 4 pendant edges of $\mathcal{T}$ (whose associated interior edges are non-adjacent) and then suppressing the degree-two vertices contains five internal subtrees isomorphic to the 3-star.

In conclusion, $\widetilde{\mathcal{T}}$ satisfies conditions~(i) and (ii). Therefore, from the previous analysis, there exists a set $\mathcal{C}$ of four 5-state characters such that $\mathcal{C}$ defines $\mathcal{T}$.
\end{proof}

For the base cases with 16 $\sim$ 18 leaves, we then adopt a leaf extension strategy --- by gradually adding leaves to expand the original tree structure into a phylogenetic tree with 19 leaves. When extending the leaves, we need to pay attention to the relationship between adjacent interior edges and characters, that is, two adjacent interior edges cannot be distinguished by the same character in a character set.

\begin{Lem}\emph{\cite{Huber2024}}\label{basic_distinguish}
Let $\mathcal{T}$ be a phylogenetic $X$-tree and $\mathcal{C}$ be a set of characters on $X$ that distinguishes $\mathcal{T}$. Then two adjacent interior edges of $\mathcal{T}$ cannot be distinguished by the same character in $\mathcal{C}$.
\end{Lem}

\begin{Cor}\label{pre_n=16_17_18}
Let $\mathcal{T}$ be a phylogenetic $X$-tree with 16, 17, or 18 leaves. Then there exists a set $\mathcal{C}$ of four 5-state characters such that $\mathcal{C}$ defines $\mathcal{T}$.
\end{Cor}
\begin{proof}
Let $n = |X|$. The case $n = 18$ is simpler than $n = 17$ and is thus omitted. First, consider the case $n = 17$. Let $\mathcal{T}$ be a phylogenetic $X$-tree with $X = \{1, 2, \ldots, 17\}$. We extend $\mathcal{T}$ to a phylogenetic $X'$-tree $\mathcal{T}'$ where $X' = \{1, 2, \ldots, 19\}$.

\begin{figure}[H]
  \centering
  \subfigure[$\mathcal{T}$]
  {
   \begin{minipage}[b]{.45\linewidth}
     \centering
     \includegraphics[width=0.60\textwidth]{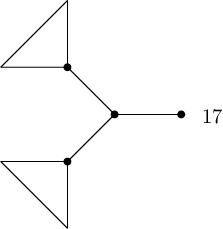}
   \end{minipage}
  }
  \subfigure[$\mathcal{T}'$]
  {
   \begin{minipage}[b]{.45\linewidth}
     \centering
     \includegraphics[width=0.75\textwidth]{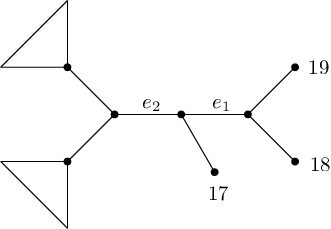}
   \end{minipage}
  }
  \caption{Extension of a phylogenetic tree with 17 leaves to 19 leaves.}
  \label{fig:n=17}
\end{figure}

By Corollary~\ref{pre_n=19_all}, there exists a set $\mathcal{C}'$ of four 5-state characters that defines $\mathcal{T}'$. According to Lemma~\ref{basic_distinguish}, edges $e_{1}$ and $e_{2}$ cannot be distinguished by the same character in $\mathcal{C}'$. Let $\mathcal{C}' = \{\chi_{1}', \chi_{2}', \chi_{3}', \chi_{4}'\}$, where $e_{1}$ is distinguished by $\chi_{1}'$ and $e_{2}$ by $\chi_{2}'$. Assume:
\[
\begin{aligned}
\pi(\chi_{1}')&=\{\{A_{11}\},\{A_{12}\},\{17,A_{13}\},\{18,19\}\},\\
\pi(\chi_{2}')&=\{\{A_{21}\},\{A_{22}\},\{A_{23}\},\{17,18,19\}\},\\
\pi(\chi_{3}')&=\{\{A_{31}\},\{A_{32}\},\{A_{33}\},\{17,18,19,A_{34}\}\}, \\
\pi(\chi_{4}')&=\{\{A_{41}\},\{A_{42}\},\{A_{43}\},\{17,18,19,A_{44}\}\}. 
\end{aligned}
\]
Remove $\{18, 19\}$ from $\pi(\chi_{1}')$ to obtain $\chi_{1}$, and replace $\{17, 18, 19\}$ in $\pi(\chi_{2}')$ with $\{17\}$ to obtain $\chi_{2}$. Similarly, remove 18 and 19 from $\pi(\chi_{3}')$ and $\pi(\chi_{4}')$ to form $\mathcal{C} = \{\chi_{1}, \chi_{2}, \chi_{3}, \chi_{4}\}$.

By construction in Fig.~\ref{fig:n=17}, $\mathcal{C}$ is convex on $\mathcal{T}$, thus defining $\mathcal{T}$. Otherwise, if there exists another $\widetilde{\mathcal{T}}$ on which $\mathcal{C}$ is convex, extending $\widetilde{\mathcal{T}}$ as in Fig.~\ref{fig:n=17} would yield a phylogenetic $X'$-tree on which $\mathcal{C}'$ is convex but not isomorphic to $\mathcal{T}'$, contradicting $\mathcal{C}'$ defining $\mathcal{T}'$.

Next, consider $n = 16$. Let $\mathcal{T}$ be a phylogenetic $X$-tree with $X = \{1, 2, \ldots, 16\}$. Extend $\mathcal{T}$ to a phylogenetic $X'$-tree $\mathcal{T}'$ with $X' = \{1, 2, \ldots, 19\}$.

\begin{figure}[H]
  \centering
  \subfigure[$\mathcal{T}$]
  {
   \begin{minipage}[b]{.45\linewidth}
     \centering
     \includegraphics[width=0.60\textwidth]{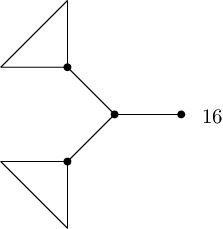}
   \end{minipage}
  }
  \subfigure[$\mathcal{T}'$]
  {
   \begin{minipage}[b]{.45\linewidth}
     \centering
     \includegraphics[width=0.85\textwidth]{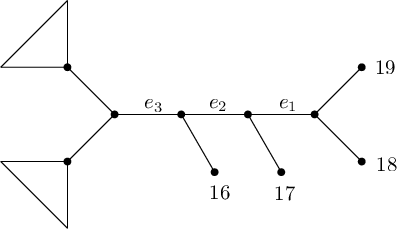}
   \end{minipage}
  }
  \caption{Extension of a phylogenetic tree with 16 leaves to 19 leaves.}
  \label{fig:n=16}
\end{figure}

By Corollary~\ref{pre_n=19_all}, there exists a set $\mathcal{C}'$ of four 5-state characters defining $\mathcal{T}'$. According to Lemma~\ref{basic_distinguish}, $e_{1}$ and $e_{2}$, as well as $e_{2}$ and $e_{3}$, cannot be distinguished by the same character in $\mathcal{C}'$. We consider two cases based on whether $e_{1}$ and $e_{3}$ are distinguished by the same character:

(1) $e_{1}$ and $e_{3}$ are distinguished by the same character in $\mathcal{C}'$. Let $\mathcal{C}' = \{\chi_{1}', \chi_{2}', \chi_{3}', \chi_{4}'\}$, where $e_{1}$ and $e_{3}$ are distinguished by $\chi_{1}'$ and $e_{2}$ by $\chi_{2}'$. Assume:
\[
\begin{aligned}
\pi(\chi_{1}') &= \{\{B_{11}\},\{B_{12}\},\{16,17\},\{18,19\}\}, \\
\pi(\chi_{2}') &= \{\{B_{21}\},\{B_{22}\},\{16,B_{23}\},\{17,18,19\}\}, \\
\pi(\chi_{3}') &= \{\{B_{31}\},\{B_{32}\},\{B_{33}\},\{16,17,18,19,B_{34}\}\}, \\
\pi(\chi_{4}') &= \{\{B_{41}\},\{B_{42}\},\{B_{43}\},\{16,17,18,19,B_{44}\}\}.
\end{aligned}
\]
Remove $\{17, 18, 19\}$ from $\pi(\chi_{2}')$ to obtain $\chi_{2}$, and remove 17, 18, 19 from $\pi(\chi_{1}')$, $\pi(\chi_{3}')$, and $\pi(\chi_{4}')$ to form $\mathcal{C} = \{\chi_{1}, \chi_{2}, \chi_{3}, \chi_{4}\}$.

(2) $e_{1}$ and $e_{3}$ are distinguished by different characters in $\mathcal{C}'$. Let $\mathcal{C}' = \{\chi_{1}', \chi_{2}', \chi_{3}', \chi_{4}'\}$, where $e_{1}$, $e_{2}$, and $e_{3}$ are distinguished by $\chi_{1}'$, $\chi_{2}'$, and $\chi_{3}'$ respectively. Assume:
\[
\begin{aligned}
\pi(\chi_{1}') &= \{\{C_{11}\},\{C_{12}\},\{16,17,C_{13}\},\{18,19\}\}, \\
\pi(\chi_{2}') &= \{\{C_{21}\},\{C_{22}\},\{16,C_{23}\},\{17,18,19\}\}, \\
\pi(\chi_{3}') &= \{\{C_{31}\},\{C_{32}\},\{C_{33}\},\{16,17,18,19\}\}, \\
\pi(\chi_{4}') &= \{\{C_{41}\},\{C_{42}\},\{C_{43}\},\{16,17,18,19,C_{44}\}\}.
\end{aligned}
\]
Remove $\{17, 18, 19\}$ from $\pi(\chi_{2}')$ to obtain $\chi_{2}$, and remove 17, 18, 19 from $\pi(\chi_{1}')$, $\pi(\chi_{3}')$, and $\pi(\chi_{4}')$ to form $\mathcal{C} = \{\chi_{1}, \chi_{2}, \chi_{3}, \chi_{4}\}$.

By construction in Fig.~\ref{fig:n=17}, $\mathcal{C}$ is convex on $\mathcal{T}$, thus defining $\mathcal{T}$. Otherwise, extending an alternative $\widetilde{\mathcal{T}}$ would contradict $\mathcal{C}'$ defining $\mathcal{T}'$.
\end{proof}

\subsection{Special Structure Phylogenetic Trees}\label{sec:Y-like}

This section discusses a class of special structure phylogenetic trees (Y-like trees). Through constructive proofs, their key property is revealed: for any such tree $\mathcal{T}$, there exists a set $\mathcal{C}$ of 5-state characters such that $\mathcal{C}$ defines $\mathcal{T}$, and the set cardinality satisfies $|\mathcal{C}| = \lceil \frac{n-3}{4} \rceil$. This conclusion lays the theoretical foundation for the recursive steps in subsequent inductive proofs, where the induction only requires the initial condition of $n \geq 20$. 

\begin{figure*}[htbp]
  \centering
  \subfigure[]
  {
   \begin{minipage}[b]{.45\linewidth}
     \centering
     \includegraphics[width=0.90\textwidth]{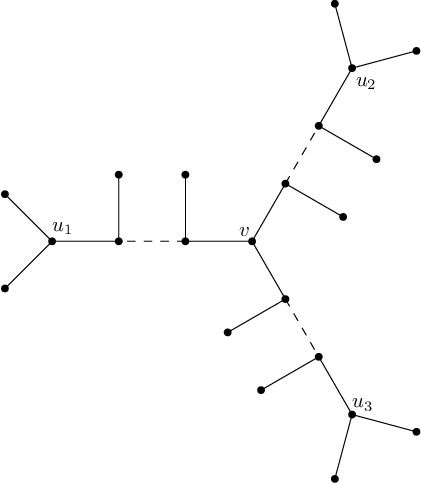}
   \end{minipage}
  }
  \subfigure[]
  {
   \begin{minipage}[b]{.45\linewidth}
     \centering
     \includegraphics[width=0.90\textwidth]{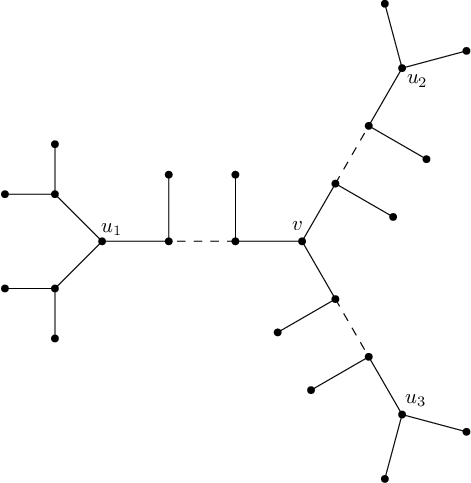}
   \end{minipage}
  }
  \subfigure[]
  {
   \begin{minipage}[b]{.45\linewidth}
     \centering
     \includegraphics[width=0.90\textwidth]{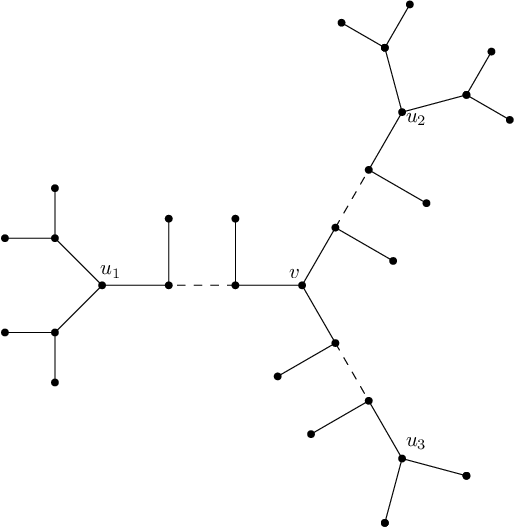}
   \end{minipage}
  }
  \subfigure[]
  {
   \begin{minipage}[b]{.45\linewidth}
     \centering
     \includegraphics[width=0.90\textwidth]{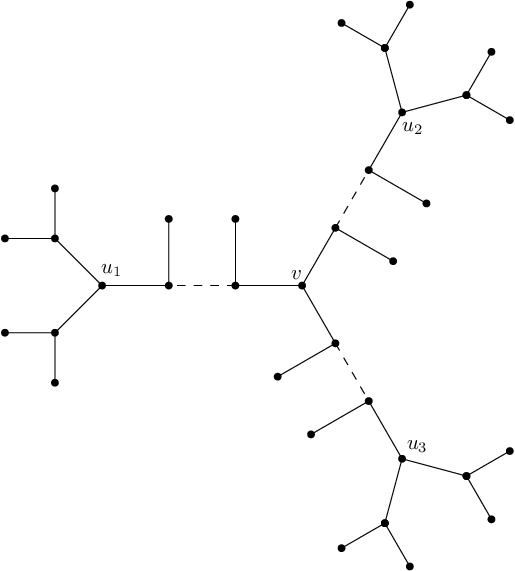}
   \end{minipage}
  }
  \caption{Y-like trees.}
  \label{fig:starlike}
\end{figure*}

As a special case of phylogenetic trees, Y-like trees have topological structures satisfying strict constraints: the tree $\mathcal{T}$ does not contain four cherries, and any two cherries must be separated by at least three interior edges. Based on topological classification by the number of cherries, we conduct the following systematic analysis:

(1) When $\mathcal{T}$ has only two cherries, $\mathcal{T}$ is a caterpillar, as shown in Fig.~\ref{fig:caterpillar}.

(2) When $\mathcal{T}$ has three cherries, the structure of $\mathcal{T}$ is like three paths ending in cherries meeting at a vertex, as shown in Fig.~\ref{fig:starlike}(a).

(3) When $\mathcal{T}$ has four cherries, since $\mathcal{T}$ does not contain four cherries with any two separated by at least three interior edges, there must exist two cherries not separated by at least three interior edges, so the structure of $\mathcal{T}$ is as shown in Fig.~\ref{fig:starlike}(b).

(4) When $\mathcal{T}$ has five cherries, the structure of $\mathcal{T}$ is as shown in Fig.~\ref{fig:starlike}(c).

(5) When $\mathcal{T}$ has six cherries, the structure of $\mathcal{T}$ is as shown in Fig.~\ref{fig:starlike}(d).

(6) When $\mathcal{T}$ has at least seven cherries, $\mathcal{T}$ must contain four cherries where any two are separated by at least three interior edges. The reason is as follows: Starting from Fig.~\ref{fig:starlike}(d), when adding a new cherry, select one of the two cherries separated by two interior edges in Fig.~\ref{fig:starlike}(d) and the new cherry, resulting in four cherries with any two separated by at least three interior edges. Therefore, the upper limit on the number of cherries in $\mathcal{T}$ is six.

Considering the structure of such trees, based on the above analysis, phylogenetic trees without four cherries and with any two cherries separated by at least three interior edges are structurally similar to the letter Y, so we call such trees \textit{Y-like trees}. The intersection vertex of the Y-shape in a Y-like tree $\mathcal{T}$ is called the center of $\mathcal{T}$, and the end vertices of the Y-shape are called the Y-ends of $\mathcal{T}$. In Fig.~\ref{fig:starlike}, the center of $\mathcal{T}$ is $v$, and the Y-ends of $\mathcal{T}$ are $u_1$, $u_2$, and $u_3$.

\begin{Lem}\label{main_lem}
Let $\mathcal{T}$ be a phylogenetic $X$-tree with $n = |X| \geq 20$. If $\mathcal{T}$ is a Y-like tree, then there exists a set $\mathcal{C}$ of 5-state characters such that $\mathcal{C}$ defines $\mathcal{T}$, where $|\mathcal{C}| = \lceil \frac{n-3}{4} \rceil$.
\end{Lem}
\begin{proof}
Based on previous analysis, when $n = |X| \geq 20$, consider the number of interior edges in the longest path of $\mathcal{T}$. When $\mathcal{T}$ has the structure of Fig.~\ref{fig:starlike}(d) and the remaining leaves (excluding those in cherries) are evenly distributed on the three paths, the number of interior edges in the longest path is minimized. In this case, the number of interior edges in the longest path is $\geq 10$. Let the center of $\mathcal{T}$ be $v$, and the Y-ends of $\mathcal{T}$ be $u_1$, $u_2$, and $u_3$, with $u_1$ and $u_2$ on the longest path.

Take 8 interior edges on the longest path, denoted as $e_1, e_2, \ldots, e_8$, and let $E_0 = \{e_1, e_2, \ldots, e_8\}$. $E_0$ must include the two edges adjacent to $v$ and one edge adjacent to each of $u_1$ and $u_2$. Let character $\chi_1$ be displayed by $\{e_1, e_3, e_5, e_7\}$, and character $\chi_2$ be displayed by $\{e_2, e_4, e_6, e_8\}$. Let $\mathcal{C}_0 = \{\chi_1, \chi_2\}$. By Lemma~\ref{basic_lem_1}, $\mathcal{C}_0$ can derive $\sigma_{e_1}, \sigma_{e_2}, \ldots, \sigma_{e_8}$, where $\chi_{e_1}, \chi_{e_2}, \ldots, \chi_{e_8}$ are the corresponding characters.

We now continue to iteratively construct the remaining characters. Let $i = 1$. In the $i$-th step, select $r-1$ interior edges of $\mathcal{T}$ to form a set $F_i$ such that each edge in $F_i$ is in a different branch of $\mathcal{T} \setminus E_{i-1}$. Define character $\chi_i$ as displayed by $F_i$. Let $E_i = E_{i-1} \cup F_i$. By induction and Lemma~\ref{basic_lem_2}, $\mathcal{C}_i = \mathcal{C}_0 \cup \{\chi_i\}$ can derive $\{\sigma_e : e \in E_i\}$. Increment $i$ by 1 and repeat.

Note that if all edges adjacent to $u_3$ are interior edges, then $F_1$ must include the interior edge adjacent to $u_3$ but not adjacent to a pendant edge. If the longest path of $\mathcal{T}$ has 10 interior edges, and since $n \geq 20$, $n$ can only be 20 or 21 in this case. Thus, a path with two pairs of cherries at both ends can be selected as the longest path, and $F_1$ can select only 3 interior edges. When $n$ is small and $n-3$ is not a multiple of 4, $F_i$ can appropriately select only 3 edges.

To ensure we do not deplete the supply of interior edges in different branches prematurely, when selecting edges for $F_i$, we always choose the 4 branches with the most interior edges in $\mathcal{T} \setminus E_{i-1}$ and select edges such that removing the edge results in two branches with similar numbers of interior edges. Through this operation, based on the analyzed phylogenetic tree structure, we can continue this process until the last step. Let the last step be $l$. When $n-3$ is not a multiple of 4, fewer than 4 interior edges will remain unselected in the $l$-th step. However, these edges are in different branches of $\mathcal{T} \setminus F_{l-1}$ and are used to generate the final character $\chi_l$, which is defined by the unselected interior edges of $\mathcal{T}$. We conclude with the 5-state character set $\mathcal{C}_l$, which can derive all non-trivial $X$-splits of $\mathcal{T}$. By Theorem~\ref{basic_thm_judgment_2}, $\mathcal{C}_l$ defines $\mathcal{T}$, where $l = \lceil \frac{n-3}{4} \rceil$.

Therefore, if $\mathcal{T}$ is a Y-like tree, a set $\mathcal{C}$ of 5-state characters can be constructed such that $\mathcal{C}$ defines $\mathcal{T}$, with $|\mathcal{C}| = \lceil \frac{n-3}{4} \rceil$.
\end{proof}

\subsection{Proof of Theorem~\ref{main_thm}}\label{sec:main-proof}

Theorem~\ref{main_thm} requires proving that for all $n \geq 16$, if $\mathcal{T}$ is a phylogenetic $X$-tree on $\{1,2,\ldots,n\}$, then there exists a set $\mathcal{C}$ of 5-state characters such that $\mathcal{C}$ defines $\mathcal{T}$, where $|\mathcal{C}| = \lceil \frac{n-3}{4} \rceil$.

This builds on the inductive strategy used for \(r=3\) [where \(|C| = \lceil \frac{n-3}{2} \rceil\) for \(n \geq 8\)] but requires more complex base case verification.
We use mathematical induction on $n$.

When $n \in \{16,17,18,19\}$, the theorem holds by Corollary~\ref{pre_n=19_all} and Corollary~\ref{pre_n=16_17_18}.

When $n \geq 20$, by Lemma~\ref{main_lem}, if $\mathcal{T}$ is a Y-like tree, then there exists a set $\mathcal{C}$ of 5-state characters such that $\mathcal{C}$ defines $\mathcal{T}$, with $|\mathcal{C}| = \lceil \frac{n-3}{4} \rceil$. Therefore, we only need to consider the case where $\mathcal{T}$ contains four cherries, with any two cherries separated by at least three interior edges. Without loss of generality, assume the labels of these four pairs of cherries are $\{n-7,n-3\}$, $\{n-6,n-2\}$, $\{n-5,n-1\}$, and $\{n-4,n\}$. Let $\mathcal{T}_{n-4}$ be the phylogenetic tree obtained by deleting the leaves labeled $n-3$, $n-2$, $n-1$, $n$ from $\mathcal{T}$ and suppressing degree-two vertices. By induction, there exists a set $\mathcal{C}_{n-4}$ of 5-state characters that defines $\mathcal{T}_{n-4}$, where $|\mathcal{C}_{n-4}| = \lceil \frac{n-7}{4} \rceil$.

At this point, $\mathcal{T}_{n-4}$ is a restriction of $\mathcal{T}$, and the interior edge set $F$ of $\mathcal{T}$ is $\mathcal{T}_{n-4}$-representable. By Lemma~\ref{basic_lem_3}, the set of characters $(\mathcal{C}_{n-4})_F$ on $\{1,2,\ldots,n\}$ can derive the $X$-splits induced by edges in $F$. Let $\chi$ be a character on $\{1,2,\ldots,n\}$ with $\pi(\chi) = \{\{n-7,n-3\},\{n-6,n-2\},\{n-5,n-1\},\{n-4,n\},\{1,2,\ldots,n-8\}\}$. Let $\mathcal{C}$ be the set $(\mathcal{C}_{n-4})_F \cup \{\chi\}$, which is a set of 5-state characters on $\{1,2,\ldots,n\}$, and $|\mathcal{C}| = \lceil \frac{n-3}{4} \rceil$. Since any two cherries in $\{n-7,n-3\}$, $\{n-6,n-2\}$, $\{n-5,n-1\}$, $\{n-4,n\}$ are separated by at least three interior edges, by Lemma~\ref{basic_lem_2}, $\mathcal{C}$ defines $\mathcal{T}$. This completes the proof.




\section{Conclusion and Prospects}\label{sec:conclusion}
This paper systematically explores the character definition theory for binary phylogenetic trees, focusing on determining the minimum leaf threshold $n_{r}$ for $r$-state character sets. 
Aiming at the threshold gaps in existing theories for $r=5$, this paper achieves theoretical breakthroughs through innovative methods.

For the case of \(r=3\), our prior research\cite{wang2025} has shown that when \(n \geq 8\), a 3-state character set with \(|C| = \lceil \frac{n-3}{2} \rceil\) uniquely defines phylogenetic trees, establishing \(n_3=8\). In this paper, we extend the framework to $r=5$ and rigorously prove that $n_5=16$. 
This result is obtained through an enhanced inductive approach, which includes base case verification ($n=16$–$19$), structural dimension reduction via 15-leaf restriction trees to exclude specific isomorphic substructures, special treatment of Y-like trees, and recursive induction employing the four-leaf deletion method. 
This study establishes, for the first time, the minimum leaf threshold for $r=5$ state character sets, providing new theoretical insight for high-dimensional character theory.

However, the research still has certain limitations. Firstly, traditional induction faces efficiency bottlenecks for $r \geq 6$, mainly due to the exponential complexity of base case verification and the dimensional sensitivity of recursive tools. Secondly, the determination of the threshold $n_{r}$ relies on manual structural analysis, lacking universal theoretical support. To address these issues, future research can advance in three directions: first, exploring new methods for character construction driven by algebraic topology or deep learning to break through the technical barriers of high-dimensional induction; second, developing threshold prediction models based on graph neural networks to establish the functional relationship between $r$ and $n_{r}$; third, expanding the application verification of the theory in bioinformatics, network topology authentication, and other fields to promote interdisciplinary integration.

The character construction theory for phylogenetic trees still has broad exploration space: firstly, the relationship between the maximum branching degree of non-binary trees and the scale of character sets needs in-depth study; secondly, rooted trees may reduce the cardinality of character sets due to the determinacy of roots, requiring quantification of the impact of root  positions; thirdly, the development of efficient algorithms (such as parallel computing or quantum optimization) is urgently needed to support large-scale character verification.

\backmatter

\addcontentsline{toc}{section}{Acknowledgements}
\bmhead{Acknowledgements}

This work was partly supported by the Fundamental Research Funds for the CCNU and partly by the Open Research Fund of Hubei Key Laboratory of Mathematical Sciences, Central China Normal University, Wuhan 430079, P. R. China. We sincerely appreciate the financial support, which enabled the successful completion of this study.

\section*{Declarations}
 \textbf{Conflict of interest}~~ The authors declare that they have no Conflict of interest. \\
 \\
 \textbf{OpenAccess} This article is licensed under a Creative Commons Attribution 4.0 International License,which permits use, sharing, adaptation, distribution and reproduction in any medium or format, as long as you give appropriate credit to the original author(s) and the source, provide a link to the Creative Commons licence, and indicate if changes were made. The images or other third party material in this article are included in the article’s Creative Commons licence, unless indicated otherwise in a credit line to the material. If material is not included in the article’s Creative Commons licence and your intended use is not permitted by statutory regulation or exceeds the permitted use, you will need to obtain permission directly from the copyright holder. To view a copy of this licence, visit \url{http://creativecommons.org/licenses/by/4.0/}





\bibliography{sn}

\end{document}